\documentclass{article}

\usepackage[preprint, nonatbib]{neurips_2023}

\usepackage[utf8]{inputenc} 
\usepackage[T1]{fontenc}    
\usepackage{hyperref}       
\usepackage{url}            
\usepackage{booktabs}       
\usepackage{amsfonts}       
\usepackage{nicefrac}       
\usepackage{microtype}      
\usepackage{xcolor}         
\usepackage{graphicx}

\usepackage{luca_macros}

\title{Learning Causal Graphs via Monotone Triangular Transport Maps}

\author{%
  \hspace{-.3cm}Sina Akbari \\
  \hspace{-.3cm}\scalebox{.8}{College of Management of Technolgy}\\
  \hspace{-.3cm}EPFL, Switzerland \\
  \hspace{-.3cm}\texttt{sina.akbari@epfl.ch} \\
  \And
  \hspace{-.4cm}Luca Ganassali \\
  \hspace{-.4cm}\scalebox{.8}{College of Management of Technolgy}\\
  \hspace{-.4cm}EPFL, Switzerland \\
  \hspace{-.4cm}\texttt{luca.ganassali@epfl.ch}\\
  \And
  \hspace{-.4cm}Negar Kiyavash \\
  \hspace{-.4cm}\scalebox{.8}{College of Management of Technolgy}\\
  \hspace{-.4cm}EPFL, Switzerland \\
  \hspace{-.4cm}\texttt{negar.kiyavash@epfl.ch} \\
}

\begin{document}

\maketitle

\begin{abstract}

We study the problem of causal structure learning from data using optimal transport (OT).
Specifically, we first provide a constraint-based method which builds upon lower-triangular monotone parametric transport maps to design conditional independence tests which are agnostic to the noise distribution. 
We provide an algorithm for causal discovery up to Markov Equivalence with no assumptions on the structural equations/noise distributions, which allows for settings with latent variables. 
Our approach also extends to score-based causal discovery by providing a novel means for defining scores. 
This allows us to uniquely recover the causal graph under additional identifiability and structural assumptions, such as additive noise or post-nonlinear models.
We provide experimental results to compare the proposed approach with the state of the art on both synthetic and real-world datasets.
\end{abstract}

\vspace{-.2cm}
\section{Introduction}
\vspace{-.1cm}
Recovering the causal structure between the variables of a system from observational data is a coveted goal in several disciplines of science.
The importance of this task has become increasingly evident in the realm of artificial intelligence over the past few decades.
This is mainly because a clear understanding of the causal structure in data can greatly enhance predictions of variables under external manipulations, and eliminate systematic biases in inference.

The existing approaches for recovering the causal mechanisms can be largely categorized into \emph{score-based} and \emph{constraint-based} methods.
Most of the existing score-based methods impose 
constraints on either the functional assignment model, or the data distribution.
For instance, they may limit the problem to linear models \cite{shimizu2006linear,seigal2022linear,zheng2018dags}, or models with additive noise \cite{hoyer2008nonlinear,ng2020role,peters2014causal,montagna2023scalable}, or restrict the data distribution to a limited class, e.g. Gaussian or discrete  \cite{kalainathan2018structural, lachapelle2019gradient,rolland2022score,montagna2023scalable}.
These methods can be sensitive to the choice of model assumptions, and may fail to recover the correct causal model if the relationships between variables are complex, or latent variables exist.
However, there is abundant evidence that information about the sparsity of the underlying causal graph improves the estimation efficiency of these methods \cite{zhang2009ica, sanchez2019estimating, shimizu2011directlingam, glymour2019review,haufe2010sparse}.
An established approach for uncovering the sparsity pattern of the graph involves conducting conditional independence tests, as commonly employed in constraint-based methods, such as PC \cite{spirtes2000causation}.
Unfortunately, conditional independence (CI) testing is only well understood -- theoretically speaking -- for either Gaussian or discrete data distributions, and proves to be inefficient in practice outside of this scope.
Despite recent progress in the study of kernel-based CI tests \cite{gretton2007kernel,chwialkowski2014kernel,zhang2012kernel,doran2014permutation}, conducting CI tests for more general data generating processes remains a daunting task.
This presents a significant challenge, since non-Gaussian continuous data is prevalent in various natural phenomena.

In this study, we employ the optimal transport (OT) framework to characterize arbitrary continuous distributions that are not necessarily Gaussian.
Through the use of this approach, we offer several advantages compared to existing approaches in the literature.
To begin with, this method provides a means to conduct conditional independence tests on continuous data with non-Gaussian distribution.
This can hence be used as a building block for every constraint-based causal discovery algorithm.
Moreover, it allows for a straightforward way to define and determine scores in the context of score-based causal discovery by characterizing the joint distribution of the variables.
We will elaborate on this point further in Section \ref{sec:anm}.

\paragraph{Related work}
To the best of our knowledge, the work presented in \cite{tu2022optimal} is the only work to date that has explored the application of optimal transport (OT) in the context of learning causal structure from data.
In this work the authors considered a two-dimensional additive noise model of the form 
\begin{equation}\label{eq:bivariate}
    (X_1,X_2) \coloneqq (U_1, f(X_1) + U_2),
\end{equation} 
where $U_1,U_2$ are independent exogenous noise variables, and sought to distinguish cause from effect among the two variables $X_1$ and $X_2$. 
The main idea in their approach comes from viewing the distribution of $(X_1,X_2)$ as the pushforward measure of that of noises $(U_1,U_2)$. 
The authors made the crucial assumption that the solution $T_{ot}$ to the following standard optimal transport problem with $L^2$ cost coincides with the structural model of Eq.~\eqref{eq:bivariate}.
\begin{equation}\label{eq:ot_a_la_ruibo}
    T_{\mathrm{ot}} := \argmin_{T \in \cM(\mu,\nu), T_1=\id} \dE_{\mu}[\Vert\mathbf{U}-T(\mathbf{U})\Vert^2],
\end{equation} 
where $\cM(\mu,\nu)$ is the set of couplings of the distribution $\mu$ of $\mathbf{U}=(U_1,U_2)$, and the distribution $\nu$ of $\mathbf{X}=(X_1,X_2)$.
Note that in \eqref{eq:ot_a_la_ruibo}, the authors impose the first coordinate of the transport map to be identity.
A simple criterion for $T_{\mathrm{ot}}$ to correspond to a cause-effect additive noise model is given by $\mathrm{div}(T_{\mathrm{ot}}-\id) = 0$, where $\mathrm{div}$ is the divergence. 
In practice, the authors rely on a conditional variance test of the form $\Var((T_{\mathrm{ot}})_{2}(\mathbf{U})-U_2 \, | \, X_1) = 0$, where $(T_{\mathrm{ot}})_{2}$ is the second coordinate of the map $T_{\mathrm{ot}}$, to test causal direction.
While the work \cite{tu2022optimal} paves the way to apply OT to causal discovery,
it remains unclear how the method can be generalized to multivariate models; for instance,
in higher dimensions ($\geq 3$),
as the proposed divergence-based criterion turns out to be only necessary, but far from sufficient.
Moreover,  OT problem in Eq.~\eqref{eq:ot_a_la_ruibo} requires the knowledge of the distribution $\mu$ which is often not available. It is unclear how robust this approach in \cite{tu2022optimal} is to noise distribution misspecification in higher dimensions.
Finally, it not clear either how to extend the approach in \cite{tu2022optimal} to models that violate the additive noise assumption.

Other works have explored the use of optimal transport framework to extract certain relations among variables from the joint distribution.
Most relevant to our study is the work \cite{spantini2018inference} which highlighted the presence of a specific type of pairwise conditional independence relations, i.e., the independence of two variables given all the rest of variables, in the Hessian information of the log density of the joint distribution. 
In line with this observation, \cite{morrison2017beyond} emphasized that these independence relations can be leveraged to recover the independence map of a Markov random field. 

Another line of work \cite{rolland2022score,montagna2023scalable} also highlights the extraction of information from the joint distribution of variables for causal discovery, but only in the case of an additive non-linear model with Gaussian noises.
The specific form of the joint distribution $\pi(\mathbf{X})$ in this case allows for reading conditional independencies pertaining to leaf nodes in the derivative of the score function $s(\mathbf{X}) := \nabla \log \pi(\mathbf{X})$. 
Such a remarkable property can be utilized to identify the leaves of a causal graph.
Applying this procedure recursively results in identifying the causal structure with linear time complexity \cite{rolland2022score,montagna2023scalable}.

\paragraph{Contributions}
Our contributions can be summarized as follows:
\begin{enumerate}[wide,labelindent=0pt,label=(\roman*)]
    \item 
    We propose a novel causal discovery method based on optimal transport (OT), designed to be agnostic to the noise distribution.
    The foundation of this method
    draws inspiration from the work by Morrison et. al. \cite{morrison2017beyond}, which originally focused on structure learning in Markov random fields.
    They utilized a parametric OT framework to infer the structure by constructing a lower triangular monotone map, denoted as $S$, between an unknown data distribution and a reference distribution, typically a standard isotropic Gaussian.
    We extend this method to causal discovery domain by incorporating additional conditional independence tests.
    It is noteworthy that our method does not rely on any assumptions regarding the structural or noise properties, and it produces the underlying causal graph up to Markov equivalence.
    Moreover, our approach is applicable in the presence of latent variables.

    \item Under additional identifiability and structural assumptions such as additive noise or post-nonlinear models, we propose further methods to recover the causal graph uniquely, based on the same construction as in $(i)$ and a notion of score given by the structural assumptions and the shape of the transport map.
    We demonstrate the dual purpose of the OT-based framework in this context: first, it facilitates the recovery of the sparsity map, and thereby enhances efficiency. Second, it offers a coherent framework for evaluating scores.

    \item We provide novel characterizations of the additive noise and post-nonlinear models, generalizing the divergence criteria of \cite{tu2022optimal} to higher dimensions.
    We show that our criteria are both necessary and sufficient for assessing whether whether data is generated from a model belonging to these two classes of SEMs.
    We demonstrate that our OT framework offers a natural and effective means for determining the validity of these criteria.
\end{enumerate}

\paragraph{Paper organization}
We define the problem in Section \ref{sec:pb_setup} and review the definitions of structural equation models and identifiability. 
We give some background on transport maps in Section \ref{sec:transport_maps}, introducing the lower-triangular monotone maps, also called Knothe-Rosenblatt maps, and argue that these maps are particularly well-adapted for causal discovery, as illustrated by Theorem \ref{th:id_with_KR} in Section \ref{sec:KR_causal}.
Section \ref{sec:mec} is dedicated to the description of our OT-based causal recovery method: we first discuss the parameterization of the learned maps, and how to extract conditional independencies from these maps. 
We then describe our algorithm, which we present as a variation of the PC algorithm \cite{spirtes2000causation}.  
Section \ref{sec:anm} describes our score-based approaches under further additive noise model or post non-linear model assumptions. 
Numerical experiments are presented in Section \ref{sec:numerical_experiments}, as well as the definition of measures of merit for our estimators. 

\vspace{-.2cm}
\section{Preliminaries}\label{sec:preliminaries}

\vspace{-0.1cm} \subsection{Problem setup}\label{sec:pb_setup}
A directed acyclic graph (DAG) is defined as $\mathcal{G}=(\mathbf{X},E)$, where $\mathbf{X}=\{X_1,\dots,X_d\}$, and $E\subseteq \mathbf{X}\times\mathbf{X}$ denote the set of vertices and edges of this graph, respectively, such that $\mathcal{G}$ cointains no directed cycle.
Each vertex $X_k\in\mathbf{X}$ represents a random variable.
For each vertex $X_k$, $\Par(X_k)$ denotes the set of it parents in $\mathcal{G}$.
~We say two DAGs are Markov equivalent if they share the same d-separation relations \cite{pearl1988probabilistic}.
Throughout this work, we assume that the random variables $\{X_1,\dots,X_d\}$ are governed by a \emph{structural equations model} (SEM) \cite{pearl2009causality},
\vspace{-0.1cm}
\begin{equation}\label{eq:sem}
    \forall \, 1 \leq k \leq d, \quad X_k \coloneqq f_k((X_\ell)_{X_\ell \in \Par(X_k)},U_k),
\end{equation} 
where $(U_k)_{1 \leq k \leq d}$ are mutually independent noise variables.
Let $\pi_{\mathbf{X}}$ denote\footnote{We drop subscript $\mathbf{X}$ whenever it is clear from context.} the probability distribution over $\mathbf{X}$ induced by the SEM defined in Eq.~\eqref{eq:sem}.
$\pi_{\mathbf{X}}$ is commonly referred to as the \emph{observational distribution} in the literature.

Causal discovery refers to the task of learning the causal graph $\mathcal{G}$  from i.i.d. samples drawn from observational distribution $\pi$.
Two causal DAGs within the same Markov equivalence class are not distinguishable from merely the observational data.
In other words, the causal graph is \emph{identifiable} up to Markov equivalence class using the observational data \cite{pearl2009causality}.
However, the causal graph may become uniquely identifiable under further assumptions.

\paragraph{Identifiability within a class of SEMs}  
Let $\cC$ be a class of SEMs, that is, a subclass of structural models of the form \eqref{eq:sem}. 
Given a distribution $\pi$, we say that the causal graph $\mathcal{G}$ is identifiable within the class $\cC$, 
if every SEM $\cS\in\cC$ that induces distribution $\pi$ yields the causal structure $\mathcal{G}$.
In other words, the causal DAG is identifiable in class $\cC$ if there is a \emph{unique} DAG $\cG$ compatible with $(i)$ the data distribution and $(ii)$ a SEM in class $\cC$.
Examples of such classes $\cC$ and some of these assumptions required for identifiability will be given in Section \ref{sec:anm}.

\begin{definition}[$\cG$-compatible orderings]\label{def:id}
Given a causal DAG $\cG$, we say that a permutation $\sigma$ of $\{ 1, \ldots, d \}$ is a \emph{($\cG$-)compatible
(causal) ordering} if for all $k,\ell$,
\begin{equation*}
    X_\ell \in \Par(X_k) \implies \sigma(\ell) < \sigma(k) \, . 
\end{equation*}
\end{definition}
Markov equivalence class of $\cG$ (a.k.a. the essential graph) together with a $\cG$-compatible ordering uniquely characterizes $\cG$.

\vspace{-0.2cm} \subsection{Background on transport maps}\label{sec:transport_maps}

Throughout this work, we assume all  probability distributions are absolutely continuous with respect to Lebesgue measure. With a slight  abuse of notation, we use the same notation to denote a distribution and its density. 

A \emph{transport map} $S$ between two distributions $\mu$ and $\nu$ in $\dR^d$ is a map $S : \dR^d \to \dR^d$ such that the pushforward of $\mu$ by $S$ is $\nu$.
In other terms $S(X) \sim \nu$ when $X \sim \mu$. Assuming $S$ is invertible, we denote by $S_{\#}\mu$ the \emph{pushforward} of $\mu$ by the map $S$, and by $S^{\#}\nu$ the \emph{pullback} of $\nu$ by $S$. These are easily obtained by a multi-dimensional change of variables as follows:
\begin{equation}\label{eq:map}
    \begin{split}
        S_{\#}\mu(\mathbf{y})&=\mu\circ S^{-1}(\mathbf{y})\:\vert\!\det(\nabla S^{-1}(\mathbf{y}))\vert,\\
        S^{\#}\nu(\mathbf{x})&=\nu\circ S(\mathbf{x})\:\vert\!\det(\nabla S(\mathbf{x}))\vert.
    \end{split}
\end{equation}
In general, there are many such transport maps $S$. 
A special class of transport maps, well suited to the problem of recovering a causal graph (see Section \ref{sec:KR_causal}), or a causal ordering, is the lower-triangular maps.  In particular, when $\mu$ and $\nu$ have positive densities with respect to Lebesgue measure, there exists a unique lower-triangular map $S$ of the following form
\begin{equation}
    S(\mathbf{x})=
    \begin{bmatrix*}[l]
        S_1(x_1)\\
        S_2(x_1,x_2)\\
        \vdots\\
        S_d(x_1,\dots,x_d)
    \end{bmatrix*},
\end{equation} which satisfies the measure transformations of Eq.~\eqref{eq:map} and such that for each component $k$, $S_k$ is strictly increasing in the last variable \cite{knothe1957contributions, rosenblatt1952remarks, bogachev2005triangular, carlier2010from}. This map is sometimes referred to as the Knothe-Rosenblatt (KR) map \cite{knothe1957contributions,rosenblatt1952remarks}.
Note that each component $S_k$ only depends on $x_1,\dots,x_k$ and the strict monotonocity of $S_k$ in $x_k$ implies that this map is invertible \footnote{Note that by this definition, the transport map defined in \cite{tu2022optimal} as the solution of problem \eqref{eq:ot_a_la_ruibo} is exactly the KR map between the distributions of  $U=(U_1,U_2)$ and $X=(X_1,X_2)$. 
}. 

We refer the interested reader to \cite{carlier2010from}, in which Knothe-Rosenblatt maps have been studied thoroughly.
In particular, KR maps are shown to be characterized as the limit of solutions of the standard optimal transport problem in the regime where
the quadratic cost becomes degenerate (see \cite{carlier2010from}, Theorem 2.1). 

\vspace{-0.2cm} \subsection{Knothe-Rosenblatt maps for causal discovery}\label{sec:KR_causal}

As discussed earlier KR maps are well suited for the problem of recovering a causal graph. 
The Theorem below, proof of which is given in Appendix \ref{apx:proofs}, makes this statement precise.
\begin{restatable}{theorem}{thmidwithkr}\label{th:id_with_KR}
Assume that random variables $\{ X_1, \ldots, X_d\}$ are governed by the SEM given in \eqref{eq:sem}, 
Moreover, assume that for all $1 \leq k \leq d$,  map $f_k$ is strictly increasing in the last variable, and that the cumulative distribution function (c.d.f) of $U_k$, denoted by $F_{U_k}$, is strictly increasing. 
    \begin{itemize}
        \item[$(i)$] Let $\sigma$ be an arbitrary $\cG$-compatible ordering. KR map $S(\sigma)$ between the distribution of $(U_{\sigma(1)}, \ldots, U_{\sigma(d)})$ and that of $(X_{\sigma(1)}, \ldots, X_{\sigma(d)})$ coincides with the SEM equations \eqref{eq:sem}, that is,
        for all $1 \leq k \leq d$,
        \begin{equation}\label{eq:th:id_with_KR}
        \hspace{-1cm}
            S(\sigma)_k(u_{\sigma(1)},\ldots,u_{{\sigma(k-1)}},u_{\sigma(k)}) = f_{\sigma(k)}\left( (S(\sigma)_{\ell}(u_{\sigma(1)}, \ldots, u_{\sigma(\ell)}))_{\ell : X_{\sigma(\ell)} \in \Par(X_{\sigma(k)})}, u_{\sigma(k)}  \right) \, .
        \end{equation}
        Moreover, KR maps corresponding to any $\cG$-compatible ordering are the same up to a permutation\footnote{Note that these permutations necessarily preserve the lower-triangular structure.}.
        \item[$(ii)$] If the causal mechanism
        is identifiable within a class $\cC$ of SEMs (Def.~\ref{def:id}), then $\sigma$ is a $\cG$-compatible ordering if and only if the KR map $S(\sigma)$ provides a SEM in class $\cC$.
    \end{itemize}
\end{restatable}

In other words, given a $\cG$-compatible ordering, KR map recovers the true underlying SEM.
Moreover, if the structure is uniquely identifiable (for instance in the case of additive noise models), computing KR maps allows us identify a $\cG$-compatible causal ordering, and hence $\cG$ itself. 

We introduce our approach for OT-based causal discovery in the two subsequent sections. This approach is comprised of two steps. 
The first step, discussed in Section \ref{sec:mec}, recovers the causal graph up to Markov equivalence, without requiring any model or structural assumptions.
The second step, details of which are given in Section \ref{sec:anm}, aims at learning a $\cG$-compatible ordering, which identifies the causal graph, under additional assumptions such as restricting the model to the ANM class.
\section{Recovering the essential graph via monotone triangular transport maps}\label{sec:mec}
\vspace{-.1cm}
\subsection{A parametrization of the transport maps} \label{sec:para_joint_map}
Henceforth, we consider the KR map from $\pi$ (or some of its marginals) to a known, smooth, log-concave source distribution $\eta$ with the same dimension. 
Throughout this work, $\eta$ will be taken to be the multivariate isotropic normal distribution\footnote{The dependence of $\eta$ on dimension will be omitted when it is evident from the context.}.
The idea behind this approach is that if the KR map $S$ from the data distribution $\pi$ to $\eta$ can be estimated efficiently with finitely many samples, then $\pi$ can be simply represented as the pullback of $\eta$ by $S$, namely $\pi = S^{\#}\eta$.

In practice, we parameterize the KR map $S$ with a vector of parameters ${\bm \alpha}$.
The parameterized transport map $S_{{\bm \alpha}}$ is estimated by optimizing over the set of parameters ${\bm \alpha}$ such that the Kullback-Leibler divergence between the density $\pi$ and the pullback of the source distribution $\eta$ by the map $S_{{\bm \alpha}}$ is minimized:
\vspace{-0.2cm}
\begin{equation}\label{eq:alpha_star}
    \begin{split}
        {\bm \alpha}^*&\!=\!\argmin_{\bm\alpha} D_\KL(\pi\Vert S_{{\bm \alpha}}^{\#}\eta) 
        \!=\!\argmin_{\bm\alpha} \mathbbm{E}_\pi[\log\pi\!-\log S_{{\bm \alpha}}^{\#}\eta]
        \!\approx\! \argmax_{\bm\alpha}\frac{1}{n}\sum_{i=1}^n\log\left(S_{{\bm \alpha}}^{\#}\eta(\mathbf{x}^i)\right),
    \end{split}
\end{equation}
where $\{\mathbf{x}^i\}_{1 \leq i \leq n}$ are the i.i.d. samples of data.
Following \cite{morrison2017beyond}, an efficient way to parameterize $S_{{\bm \alpha}}$ in order to enforce both the lower-triangular shape and the monotonicity assumption is the following:
\vspace{-0.1cm}
\begin{equation}\label{eq:param_S}
    (S_{{\bm \alpha}})_k(x_1, \ldots, x_k) = c_{k, {\bm \alpha}}(x_1, \ldots, x_{k-1}) + \int_{0}^{x_k} g \circ h_{k, {\bm \alpha}}(x_1, \ldots, x_{k-1},t) \mathrm{d}t,
\end{equation} where $g:\dR \to \dR$ is a positive map, and $c_{k, {\bm \alpha}} : \dR^{k-1} \to \dR$ (resp. $h_{k, {\bm \alpha}} : \dR^{k} \to \dR$) is a linear combination of multivariate Hermite polynomials $\{\phi_s\}_{s \geq 0}$ (resp. multivariate Hermite functions $\{\psi_s\}_{s \geq 0}$).

Note that map $S_{{\bm \alpha}}$ defined in \eqref{eq:param_S} has the desired lower-triangular shape and each $(S_{{\bm \alpha}})_k$ is strictly increasing in the last variable.
The well-known fact that Hermite functions form a Hilbert basis of $L^2(\dR)$ justifies this parametrization \cite{morrison2017beyond,spantini2018inference}. The expressiveness of the model depends on the maximal degree of the Hermite polynomials/functions 
in $(c_{k, {\bm \alpha}},h_{k, {\bm \alpha}})_{1 \leq k \leq d}$. 
Moreover, since $\eta$ is log-concave and the parametrization of $S_{{\bm \alpha}}$ is linear in $\alpha$, the optimization problem \eqref{eq:alpha_star} is convex. In some recent work was  proposed to learn maps $c_{k, {\bm \alpha}}$ and $h_{k, {\bm \alpha}}$ with neural networks \cite{papamakarios2017}.  

In this work, the positive map $g$ in \eqref{eq:param_S} is always fixed as the square function.
This allows us to compute all the integrals in maps $S_{\bm \alpha}$ easily and in closed form (as they correspond to moments of Gaussian variables), 
as well as all the partial derivatives of $S_{\bm \alpha}$ -- which we will need in the sequel.

\vspace{-0.2cm} \subsection{Capturing conditional independencies in marginal densities}\label{sec:closure}
We denote the independence of $X_\ell$ and $X_k$ conditioned on $\mathbf{Z}$ by $X_\ell\indep X_k\vert \mathbf{Z}$.
Recall that $\pi$ denotes the distribution (or, density) of the data $\mathbf{X}=\{X_1,\dots,X_d\}$.
In order to explain how the conditional independencies can be read directly off the marginals of $\pi$, we need the following assumption on $\pi$.
\begin{assumption}[Smoothness \& strict positivity]\label{assum:smooth_positivity}
    Density $\pi$ 
    is positive and its second-order partial derivatives are defined everywhere.
\end{assumption}

Lemma 2 of \cite{spantini2018inference} establishes a characterization of conditional independence in terms of Hessian information of the density $\pi$.
Herein, we adapt their lemma for our purpose.
\begin{restatable}{lemma}{lemcitest}[Adapted from \cite{spantini2018inference}]
\label{lem:citest} 
    Suppose $\pi$ satisfies Assumption \ref{assum:smooth_positivity}.
    Let $\pi_\mathbf{Z}$ denote the marginal density over $\mathbf{Z}\subseteq\mathbf{X}$.
    For any two variables $X_k,X_\ell\in\mathbf{Z}$, the following equivalence holds:
    \vspace{-0.1cm}
    \begin{equation*}
        X_k \indep X_\ell \mid \mathbf{Z}\setminus\{X_k,X_\ell\} \iff \dfrac{\partial^2\log\pi_{\mathbf{Z}}}{\partial x_k\partial x_\ell}=0 \mbox{ on $\dR^{\vert\mathbf{Z}\vert}$}\, .
    \end{equation*}

\end{restatable}
Proof of Lemma \ref{lem:citest} appears in Appendix \ref{apx:proofs}. 
%


For a fixed subset $\mathbf{Z}\subseteq\mathbf{X}$, let $S_\mathbf{Z}$ be a transport map that pushes the marginal density $\pi_\mathbf{Z}$ forward to a multivariate isotropic Gaussian $\eta$.
In light of Lemma \ref{lem:citest}, 
the conditional independence relations can be determined by assessing the partial derivatives of $S_\mathbf{Z}^\#\eta$ at all points and observing if they are all zero. 
When the variables are continuous, determining whether these derivatives are zero everywhere is impractical.
Instead, 
we propose the following \emph{conditional independence score\footnote{Note that for an arbitrary real-valued function $f$ and variables $\mathbf{Z}$, 
$\dE_\pi[f(\mathbf{Z})^2]=0\iff f(\mathbf{Z})=0 \:\:\:\pi-a.s.$}}
to test the conditional independence $X_k\indep X_\ell \mid \mathbf{Z}\setminus\{X_k,X_\ell\}$:
\begin{flalign}
    \Omega^\mathbf{Z}_{k \ell} & \coloneqq \!\dE_{\pi_\mathbf{Z}}\left[ \left( \dfrac{\partial^2}{\partial x_k\partial x_\ell} \log\pi_{\mathbf{Z}}(\mathbf{z})\right)^2 \right]
     \!=\! \dE_{\pi_\mathbf{Z}}\left[ \left( \dfrac{\partial^2}{\partial x_k\partial x_\ell} \log S_{\mathbf{Z}}^\#\eta(\mathbf{z}) \right)^2 \right] 
      \!\approx\! \frac{1}{n}\!\sum_{i=1}^n\left(\dfrac{\partial^2}{\partial x_k\partial x_\ell}S_{\mathbf{Z}}^\#\eta(\mathbf{z}^i)\right)^2, \label{eq:def_omega}
\end{flalign}
where $\{\mathbf{z}^i\}_{i=1}^n$ are observed samples of $\mathbf{Z}$.
In practice, finite sample approximations of $\bm\alpha$ and $\Omega^\mathbf{Z}_{k\ell}$ could yield small but non-zero entries when the corresponding independence holds.
To deal with this issue, 
we compare the conditional independence score to a properly chosen threshold $\tau^\mathbf{z}_{k \ell}$.
The threshold is chosen in proportion to the standard deviation of $\Omega^\mathbf{Z}_{k\ell}$, driven by the objective of isolating those entries whose standard deviation renders them indistinguishable from zero.
Morrison et. al. \cite{morrison2017beyond} 
take a similar thresholding approach, albeit employing the absolute value of the partial derivatives of log density as the independence score rather than the squared form of Eq.~\eqref{eq:def_omega}.
The standard deviation of $\Omega^\mathbf{Z}_{k\ell}$ is approximated as
\[\varsigma(\Omega^\mathbf{Z}_{k\ell})\approx
\frac{1}{n}(\nabla_{\bm\alpha}\Omega^\mathbf{Z}_{k\ell})^T\Gamma({\bm\alpha})^{-1}(\nabla_{\bm\alpha}\Omega^\mathbf{Z}_{k\ell})\Big\vert_{\bm\alpha=\bm\alpha^*},\]
where $\Gamma({\bm\alpha})$ is the Fisher information matrix \cite{casella2002statistical}, and $\nabla_{\bm\alpha}\Omega^\mathbf{Z}_{k\ell}$ denotes the gradient of $\Omega^\mathbf{Z}_{k\ell}$ with respect to $\bm\alpha$.
For further details of the rationale behind this choice of thresholding and its consistency analysis, see \cite{morrison2017beyond}.

It is crucial to note that this criterion does not require any assumption on the distribution class.
SING algorithm \cite{morrison2017beyond} employs the criterion above
only for the set $\mathbf{Z}\coloneqq\mathbf{X}$ to recover the Markov random field structure.
In the context of DAGs, their approach is termed as \emph{total conditioning} (TC) by \cite{pellet2008using}, and is shown to recover the moralized graph of $\mathcal{G}$, where each vertex is adjacent to its Markov boundary \cite{pellet2008using}. In our method, the maps $S_{\mathbf{Z}}$ for every $\mathbf{Z}\subseteq\mathbf{X}$ will be parameterized in the same fashion as $S$ in \eqref{eq:param_S}, except that dimension $d$ will be replaced with a smaller $d' := \left| \mathbf{Z} \right|$.

\vspace{-0.2cm} \subsection{Description of the algorithm}
Under faithfulness assumption \cite{spirtes2000causation},
the conditional independence relations encoded in the data distribution are equivalent to d-separations in the causal DAG.
The conditional independence test developed in the previous section can therefore be employed as a module in any constraint-based causal discovery method. 
To illustrate this point, we present our method PC-OT as a variation of the PC algorithm \cite{spirtes2000causation}, summarized as Algorithm \ref{alg:pc} in Appendix \ref{apx:pc}. Our algorithm begins with estimating the two-dimensional marginals, revealing some of the independencies, and iteratively increases the size of set $\mathbf{Z}$.
Note that even in the presence of latent variables, the joint distribution of the observable variables can still be represented as a pullback measure of a standard Gaussian distribution of the same dimension.
Although this map may have nothing to do with the underlying SEM in this case, once a representation $\pi = S^{\#}\eta$ (or\footnote{In general, the marginal distribution $\pi_{\mathbf{Z}}$ for $\mathbf{Z}\subseteq\mathbf{X}$ is not modeled by a SEM with independent noises, and hence the marginal model for nodes in $\mathbf{Z}$ is equivalent to a SEM with latent variables.} $\pi_\mathbf{Z} = S_\mathbf{Z}^{\#}\eta$) is obtained, Lemma \ref{lem:citest} applies and the conditional independencies are revealed. 
In the presence of latent variables, OT-based variants of algorithms such as FCI and RFCI \cite{spirtes2000causation, colombo2012learning} can be developed analogous to Alg.~\ref{alg:pc}.

\vspace{-.2cm}
\section{Refined OT-based structure learning: going further under structural  assumptions}\label{sec:anm}
\vspace{-.1cm}

In this section, we shall undertake an analysis of our OT-based causal discovery approach under the consideration of class assumptions that enable the unique identification of the causal DAG.
The results stated in this section are applicable to any class within which the causal graph is identifiable (refer to Def.~\ref{def:id}) and membership in that class is testable.
As two illustrative examples of such classes, we discuss additive noise models (ANMs) and post-nonlinear models (PNLs) in the sequel.

\vspace{-0.2cm} \subsection{Additive noise models}
Additive noise models (ANMs) are defined as follows.
\begin{definition}[ANM]
\label{assum:ANM}
We say that the SEM of Eq.~\eqref{eq:sem} forms an additive noise model if  
\begin{equation}\label{eq:assum:ANM}
    \forall \, 1 \leq k \leq d, \quad f_k((X_\ell)_{X_\ell \in \Par(X_k)},U_k)\coloneqq g_k((X_\ell)_{X_\ell \in \Par(X_k)}) + U_k,
\end{equation} 
that is, the structural equation pertaining to any variable is additive in the corresponding noise.
\end{definition}
Note that as long as the noise variables $U_k$ have a strictly positive density, the ANMs satisfy the conditions of Theorem \ref{th:id_with_KR}.
For any ordering $\sigma$, we denote by $\pi_\sigma$ the joint distribution of $(X_{\sigma(1)}, \ldots, X_{\sigma(d)})$. 
If $\sigma$ is $\cG-$compatible, then in view of Theorem \ref{th:id_with_KR}, the KR map $S(\sigma)$ from $\pi_\sigma$ to $\eta$ is of the form 
\vspace{-0.2cm} 
\begin{equation}\label{eq:ANM_KR_map}
S(\sigma)_k(x_{\sigma(1)}, \ldots, x_{\sigma(k)}) = M_k(\sigma) \left(x_{\sigma(k)} - g_k((x_\ell)_{X_\ell \in \Par(X_{\sigma(k)})})\right), 
\end{equation} where $M_k(\sigma)$ is the strictly increasing transport map from the distribution of $U_{\sigma(k)}$ to a standard Gaussian $\cN(0,1)$. Note that, up to a one-dimensional monotonous map, the partial derivative of $S(\sigma)_k$ with respect to its last variable is a constant, and specifically equal to $1$. 
This observation constitutes a characterization of ANMs, as formalized below.
\begin{restatable}{lemma}{lemanmtest}\label{lem:ANMtest} 
    Suppose $\pi$ satisfies Assumption \ref{assum:smooth_positivity}.
    Let $\sigma$ be an ordering. Then, the following are equivalent.
    \begin{itemize}
        \item $\pi$ is induced by an ANM, and $\sigma$ is $\cG-$compatible.
        \item for all $1 \leq k \leq d$, there exists a strictly increasing map $B_k(\sigma): \dR \to \dR$ such that 
        \begin{equation}\label{eq:lem:ANMtest}
            \frac{\partial}{\partial x_{\sigma(k)}} B_k(\sigma) \circ S(\sigma)_k (x_{\sigma(1)}, \ldots, x_{\sigma(k)}) - 1 = 0 \, .
        \end{equation} 
    \end{itemize}
\end{restatable}
\vspace{-0.4cm}
Note that $B_k(\sigma)$ in Eq.~\eqref{eq:lem:ANMtest} corresponds to $M_k(\sigma)^{-1}$ in Eq.~\eqref{eq:ANM_KR_map}. In practice, for a given ordering $\sigma$, we can parameterize each map $B_k(\sigma)$ with vector ${\bm \beta_k}$, and ${\bm \beta^*_k}$ is estimated by optimizing the natural loss given by Lemma \ref{lem:ANMtest}:
\vspace{-0.2cm}
\begin{flalign}\label{eq:beta_star_k}
    {\bm \beta^*_k} & := \argmin_{\bm \beta_k} \dE_{\pi}\left[ \left| \frac{\partial}{\partial x_{\sigma(k)}} [B_k(\sigma)]_{{\bm \beta_k}} \circ (S(\sigma)_{\bm \alpha^*})_k(X_{\sigma(1)},\ldots,X_{\sigma(k)}) -1 \right| \right] \nonumber \\
    & \approx \argmin_{\bm \beta_k} \,  [\mathrm{ANMloss}_{k}(\sigma,\mathbf{x})]({\bm \beta_k}),
\end{flalign} where
\vspace{-0.3cm} 
\begin{equation}\label{eq:ANM_loss_k_beta}
    [\mathrm{ANMloss}_{k}(\sigma,\mathbf{x})]({\bm \beta_k}) := \sum_{i=1}^{n} \left| \frac{\partial}{\partial x_{\sigma(k)}} [B_k(\sigma)]_{{\bm \beta_k}} \circ (S(\sigma)_{\bm \alpha^*})_k(x^i_{\sigma(1)},\ldots,x^i_{\sigma(k)}) -1 \right|, 
\end{equation}
and (we recall) $\mathbf{x}^1, \ldots, \mathbf{x}^n$ are observed samples. 
Note that we used $S(\sigma)_{\bm \alpha^*}$ with ${\bm \alpha^*}$ being the solution of the optimization problem \eqref{eq:alpha_star}, as introduced in Section \ref{sec:mec}.

An efficient way to parameterize $[B_k(\sigma)]_{{\bm \beta_k}}$ 
is as follows:
\begin{equation}\label{eq:param_beta_k}
[B_k(\sigma)]_{{\bm \beta_k}}(u)  := \int_{0}^{u} g \circ b_{k,{\bm \beta_k}}(t) \mathrm{d}t ,
    \end{equation} where $g:\dR \to \dR$ is a positive (the quadratic function in our case), and $b_{k,{\bm \beta_k}} : \dR \to \dR$ is a linear combinations of Hermite functions $\{\psi_s\}_{s \geq 0}$. Note that the parameterization in \eqref{eq:param_beta_k} enforces strict monotonicity of $u \mapsto [B_k(\sigma)]_{{\bm \beta_k}}(u)$.

Note that given the Markov equivalence class, and a $\cG$-compatible ordering, the causal graph is uniquely identified.
Under identifiability assumptions for ANMs \cite{hoyer2008nonlinear}, every $\cG$-compatible ordering is consistent with the true underlying causal order.
As such, identifying one $\cG$-compatible ordering suffices to recover the causal DAG.
On account of Lemma \ref{lem:ANMtest}, we devise a method to decide $\cG$-compatibility of an ordering under ANM assumption.
To this end, we compute the \emph{ANM loss} corresponding to an ordering, introduced subsequently.

\paragraph{ANM loss.} 
The \emph{ANM loss} of an ordering $\sigma$, parameterized by ${\bm \gamma} \in (\dR_{>0})^d$ is defined as
\vspace{-0.12cm}
\begin{equation}\label{eq:ANM_loss}
    \mathrm{ANMloss}_{\bm \gamma}(\sigma,\mathbf{x}) := \sum_{k=1}^{d} \gamma_k [\mathrm{ANMloss}_{k}(\sigma,\mathbf{x})]({\bm \beta^*_k}),
\end{equation} 
where ${\bm \alpha^*}$, ${\bm \beta^*_k}$ and $\mathrm{ANMloss}_{k}(\sigma,\mathbf{x})$ are defined in \eqref{eq:alpha_star}, \eqref{eq:beta_star_k} and \eqref{eq:ANM_loss_k_beta}. 
Evidently, an ordering $\sigma$ is $\cG$-compatible if and only its ANM loss is zero.
In practice, we choose the ordering with the lowest ANM loss. Algorithm \ref{alg:anm} illustrates this procedure, which takes as input the essential graph of $\cG$\footnote{For instance, Algorithm \ref{alg:pc} can be utilized to recover the essential graph from data.}.


\paragraph{Possible orderings} Given an essential graph $\hat{\cG}$, we need to test every possible causal graph $\cH$ in the Markov equivalence class $\cM(\hat{\cG})$ of $\hat{\cG}$. For any such graph $\cH$, choose an arbitrary $\cH-$compatible ordering $\sigma_\cH$.
We define the set of \emph{possible orderings}, denoted $\Sigma(\hat{G})$, as follows:
\begin{equation}\label{eq:def_possible_orderings}
    \Sigma(\hat{G}) := \{ \sigma_\cH, \, \cH \in \cM(\hat{\cG}) \} \, .
    \end{equation} Note that in general, depending on $\hat{\cG}$, $\Sigma(\hat{G})$ is way smaller than the set of permutations of $\{ 1, \ldots, d\}$, a trivial upper bound on $\Sigma(\hat{G})$ is $2^{e(\hat{G})}$, where $e(\hat{G})$ is the number of edges in $\hat{G}$

\setlength{\textfloatsep}{8pt}
\begin{algorithm}
\caption{ANM--OT}
\label{alg:anm}
\begin{algorithmic}[1]
    \Statex\textbf{input:} an essential graph $\hat{G}$ , $n$ i.i.d. samples from the observational distribution $\{\mathbf{x}^i\}_{i=1}^n\sim\pi$, scaling parameters $\gamma$ 
    \Statex\textbf{output:} a causal DAG $\mathcal{G}$ 
        \For{every possible ordering $\sigma \in \Sigma(\hat{G})$} 
        \State $\ell(\sigma, \mathbf{x}) \gets \mathrm{ANMloss}_{\bm \gamma}(\sigma, \mathbf{x})$
    \EndFor
    \State\Return the causal DAG $\mathcal{G}$  with causal ordering $\sigma_{\mathrm{opt}} \in \argmin_{\sigma \in \Sigma} \ell(\sigma, \mathbf{x})$
\end{algorithmic}
\end{algorithm}

\vspace{-0.2cm} \subsection{Post non-linear models}
Post non-linear (PNL) models  \cite{uemura2022multivariate,uemura2020estimation,zhang2015estimation}, known to be a general identifiable class of models, are defined as follows. 

\begin{definition}[PNL]
\label{assum:PNL}
We say that the SEM of Eq.~\eqref{eq:sem} forms a post-nonlinear model if
\begin{equation}\label{eq:assum:PNL}
    \forall \, 1 \leq k \leq d, \quad 
    f_k((X_\ell)_{X_\ell \in \Par(X_k)},U_k)\coloneqq h_k\left(g_k((X_\ell)_{X_\ell \in \Par(X_k)}) + U_k\right),
\end{equation}  
where the functions $h_k$ are invertible.
\end{definition}

Note that the PNL model reduces to an ANM when $h_k$ is the identity for all $k$.
That is, ANMs are a special case of PNLs.
The PNL class is identifiable if we prevent some singular functions and noise distributions \cite{zhang2009pnls}. 
We show the following characterization of PNLs.
\begin{restatable}{lemma}{lempnltest}\label{lem:PNLtest} 
    Suppose $\pi$ satisfies Assumption \ref{assum:smooth_positivity}.
    Let $\sigma$ be an ordering. The following are equivalent.
    \begin{itemize}
        \item $\pi$ is induced by a PNL model, and $\sigma$ is $\cG-$compatible.
        \item For all $1 \leq k \leq d$, there exists a strictly increasing map $B_k(\sigma): \dR \to \dR$ such that 
        \begin{equation}\label{eq:lem:PNLtest}
            \forall \, 1 \leq \ell \leq d, \, \ell \neq k, \quad \frac{\partial^2}{\partial x_{\sigma(l)} \partial x_{\sigma(k)}} B_k(\sigma) \circ S(\sigma)_k (x_{\sigma(1)}, \ldots, x_{\sigma(k)}) = 0 \, .
        \end{equation} 
    \end{itemize}
\end{restatable} As in the case of ANMs, the characterization of Lemma \ref{lem:PNLtest} can be employed to design a PNL loss for an ordering.
Due to space limitation, we postpone the corresponding derivations to Appendix \ref{apx:pnl}.
With this PNL loss at hand, Algorithm \ref{alg:anm} can then be run with replacing $\mathrm{ANMloss}$ with $\mathrm{PNLloss}$.

\vspace{-.2cm}
\section{Numerical experiments}\label{sec:numerical_experiments}
\vspace{-.1cm}
This section presents the performance of our methods through illustrative examples. 
Further numerical experiments, plots and details can be found in Appendix \ref{apx:numerical_experiments}. 
We utilized TransportMaps package \cite{tmpackage} for recovering parametric maps.
\subsection{Experiments for PC--OT}\label{sec:exp_pc_ot}

In the sequel, we compared the performance of our PC--OT method with both PC algorithm \cite{spirtes2000causation} and Grow-Shrink (GS) algorithm \cite{spirtes2000causation}, provided with conditional independence tests which are designed for Gaussian distributions (using CDT package \cite{kalainathan2019causal}). 
On the contrary, our method is equipped with the OT-based conditional independence criterion, which is agnostic to the noise distribution.

We worked with synthetic data from a SEM where the exogenous noises are non-Gaussian, details of which can be found in Appendix \ref{apx:numerical_experiments}. The underlying causal graph is represented in Figure \ref{fig:causal_graphs_experiments}b. We see on Figure \ref{fig:PCOT_vs_PC_GS} that PC--OT outperforms these two methods as soon as the number of samples is large enough. This superior performance is illustrated for the number of misoriented edges of the output as well as the overall loss taking into account missing, extra and misoriented edges. Results are averaged over $20$ tests.
See Appendix \ref{apx:numerical_experiments} for further details.

\begin{figure}
	\centering
    \hspace{-1cm}
\begin{subfigure}[c]{0.45\textwidth}
        \includegraphics[scale=0.44]{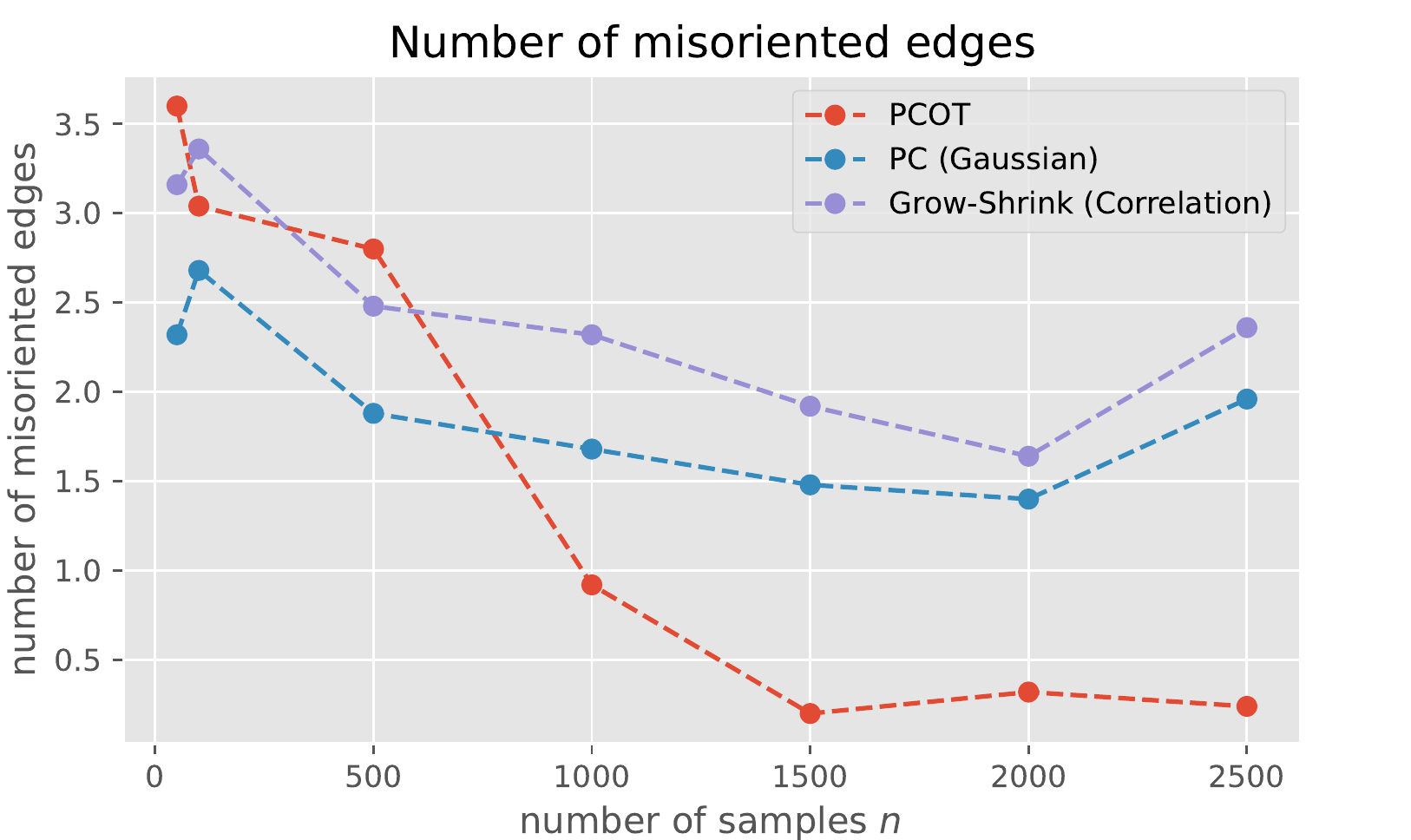}
  \subcaption{average number of misoriented edges}
  \label{fig:misorientedmain}
	\end{subfigure}
 \hspace{1.5cm}
 \begin{subfigure}[c]{0.45\textwidth}
		\includegraphics[scale=0.44]{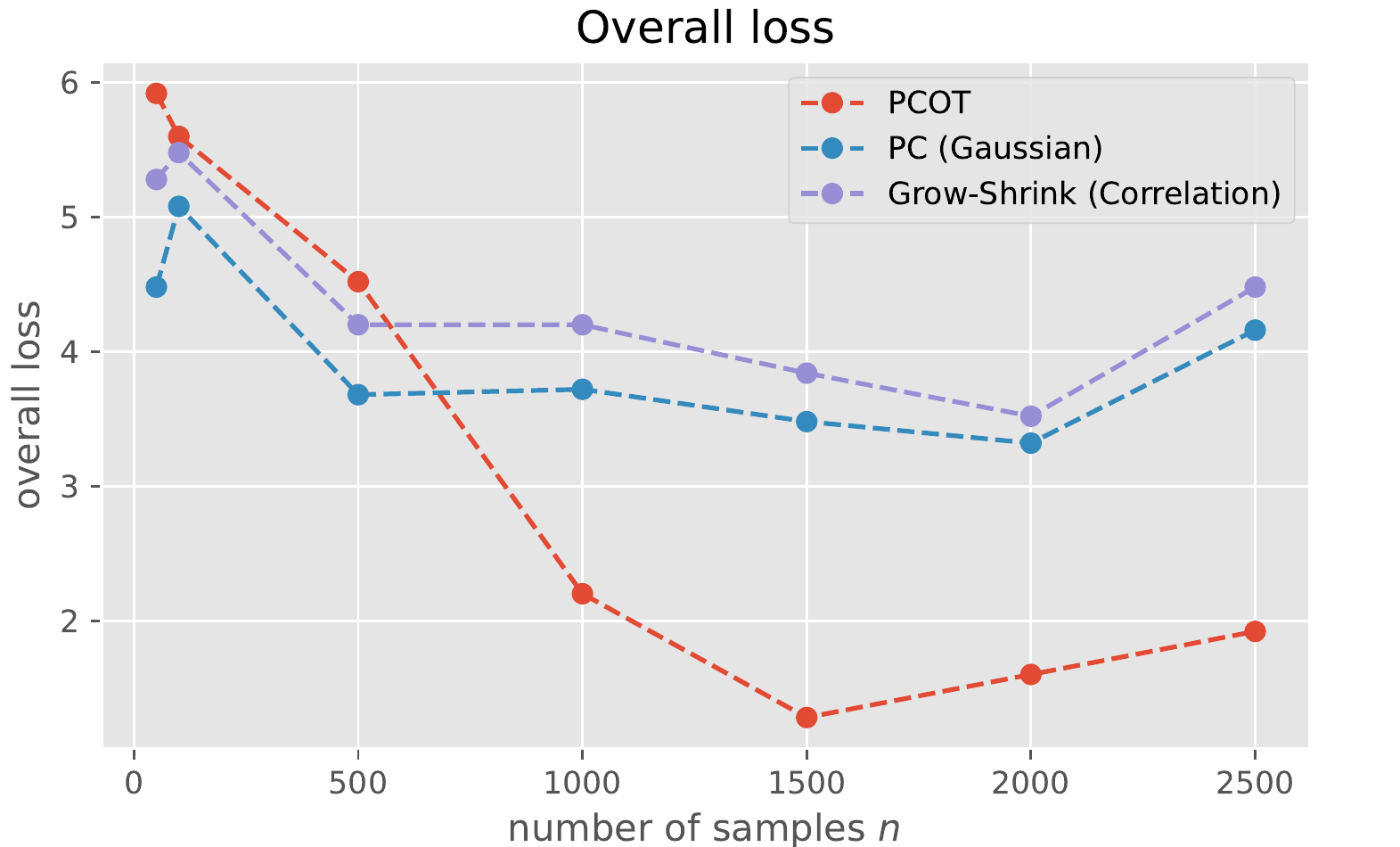}
  \subcaption{overall average loss}
  \label{fig:overall}
	\end{subfigure}
\caption{comparison of PC--OT with PC and GS algorithms.}
\label{fig:PCOT_vs_PC_GS}
\end{figure}

\subsection{Experiments for ANM-OT}

In order to illustrate the ANM-OT method, we worked with the DAG of Figure \ref{fig:causal_graphs_experiments}a. Note that the Markov equivalence class of this graph contains 4 DAGs. For each of these DAGs, we chose a compatible ordering and compared the $\mathrm{ANMloss}$ \eqref{eq:ANM_loss} of each of these orderings. Results shown in Figure \ref{fig:ANM-OT_experiments} show that the $\mathrm{ANMloss}$ of the true ordering (the rightmost one) is significantly lower than the other three. 

\begin{figure}[h]
	\centering
	\begin{subfigure}[c]{0.45\textwidth}
		\centering
        \includegraphics[scale=0.38]{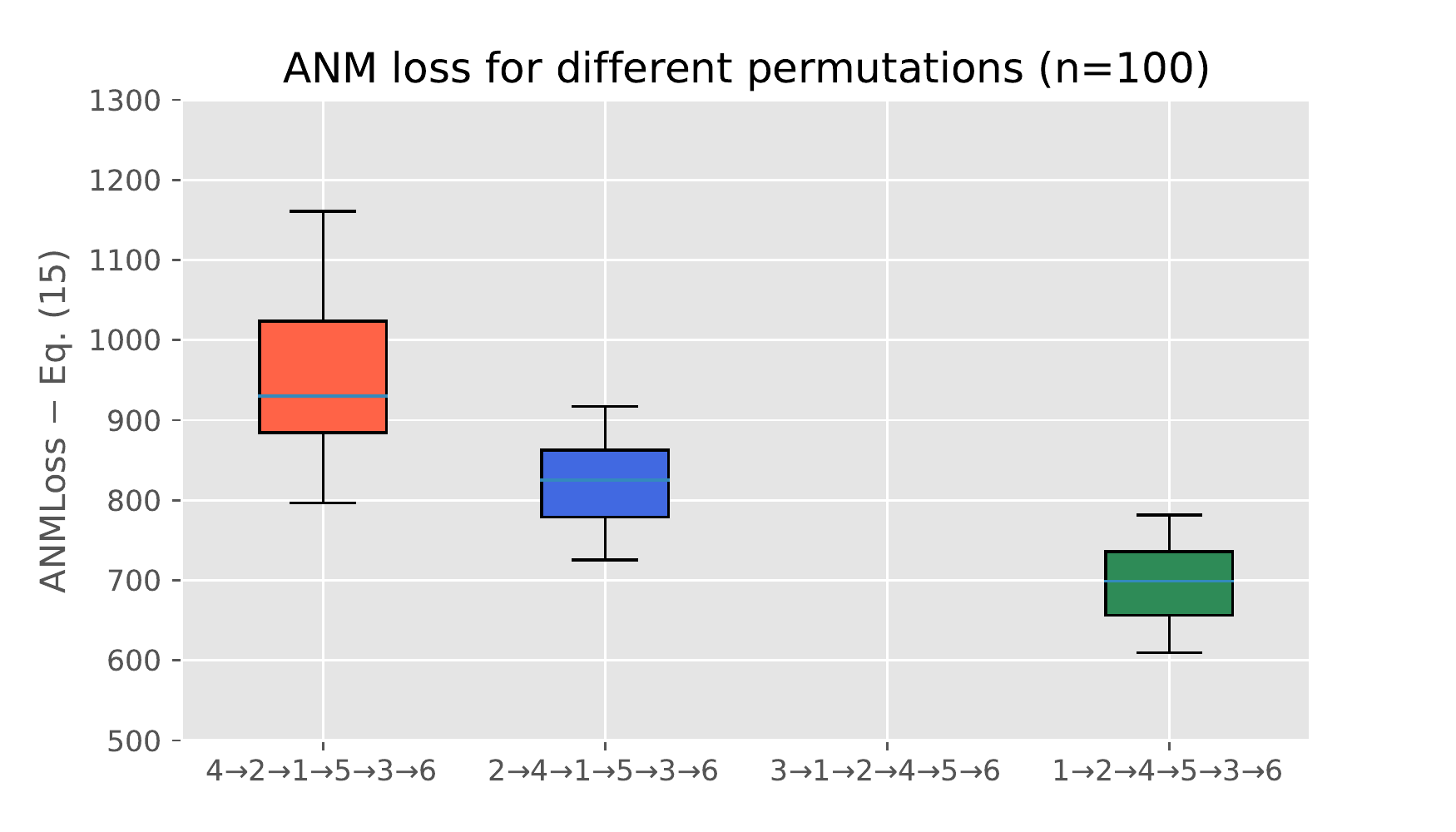}
    \label{fig:3a}
  \subcaption{ANMloss with $n=100$ samples available}
	\end{subfigure}
 \hfill
 \begin{subfigure}[c]{0.45\textwidth}
		\centering
		\includegraphics[scale=.38]{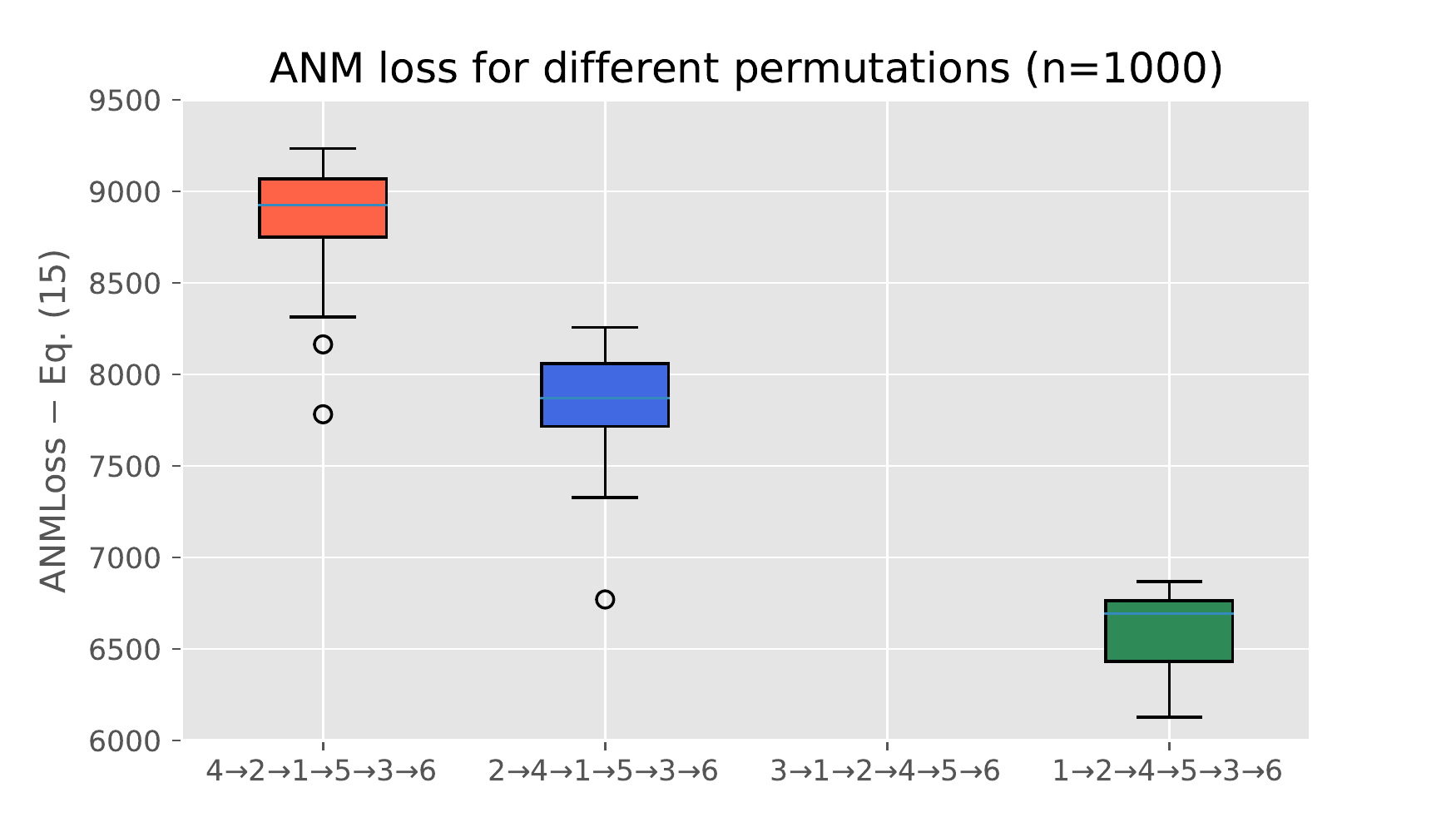}
  \label{fig:3b}
  \subcaption{ANMloss with $n=1000$ samples available}
	\end{subfigure}
\caption{ANMloss for the four different DAGs within the Markov equivalence class:
(a) with 100 samples, and (b) with 1000 samples.
For both plots, $\gamma_k=1$ was chosen for every $1\leq k\leq d$ (see Eq.~\eqref{eq:ANM_loss}.)
The plots are clipped for better visualization: the loss corresponding to ordering $3\to1\to2\to4\to5\to6$ is not included due to a large gap.}
\label{fig:ANM-OT_experiments}
\end{figure}

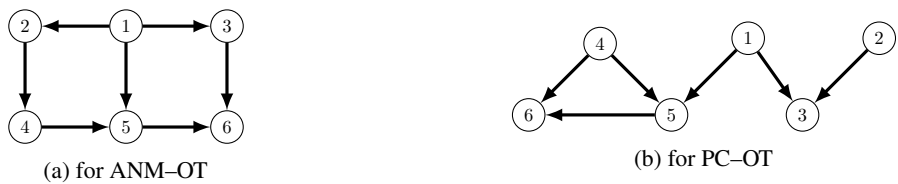
\begin{figure}[h]
	\centering
	\begin{subfigure}[c]{0.45\textwidth}
		\centering
		\begin{tikzpicture}[scale=0.6, every node/.style={scale=0.7}]
			\tikzset{vertex/.style = {shape=circle,draw,minimum size=1em}}
			\tikzset{edge/.style = {->,very thick,> = latex,sibling distance=20mm}}
			\node[vertex] (a) at (0,0) {$1$};
			\node[vertex] (b) [left  = 0.9cm of a] {$2$};
			\node[vertex] (c) [right = 0.9cm of a] {$3$};
			\node[vertex] (d) [below = 0.9cm of b] {$4$};
			\node[vertex] (e) [right = 0.9cm of d] {$5$};
			\node[vertex] (f) [right = 0.9cm of e] {$6$};
			\draw[edge] (a) to (b);
			\draw[edge] (a) to (c);
			\draw[edge] (b) to (d);
			\draw[edge] (c) to (f);
			\draw[edge] (d) to (e);
			\draw[edge] (e) to (f);
            \draw[edge] (a) to (e);
		\end{tikzpicture}
    \subcaption{for ANM--OT}
    \label{fig:exp-anm}
	\end{subfigure}
 \hfill
 \begin{subfigure}[c]{0.45\textwidth}
		\centering
		\begin{tikzpicture}[scale=0.6, every node/.style={scale=0.7}]
			\tikzset{vertex/.style = {shape=circle,draw,minimum size=1em}}
			\tikzset{edge/.style = {->,very thick,> = latex,sibling distance=20mm}}
			\node[vertex] (a) at (0,0) {$1$};
			\node[vertex] (b) [right  = 1.3cm of a] {$2$};
			\node[vertex] (c) [below left = 1cm of b] {$3$};
			\node[vertex] (e) [below left = 1cm of a] {$5$};
			\node[vertex] (d) [above left = 0.9cm of e] {$4$};
			\node[vertex] (f) [below left = 0.9cm of d] {$6$};
			\draw[edge] (a) to (c);
			\draw[edge] (b) to (c);
			\draw[edge] (a) to (e);
			\draw[edge] (d) to (e);
			\draw[edge] (d) to (f);
			\draw[edge] (e) to (f);
		\end{tikzpicture}
    \subcaption{for PC--OT}
    \label{fig:exp-pcot}
	\end{subfigure}
\caption{Underlying causal graphs in the numerical experiments.}
\label{fig:causal_graphs_experiments}
\end{figure}

\section{Concluding remarks}\label{sec:concluding_remarks}

In this work, we proposed a novel causal discovery method based on optimal transport (OT), designed to be agnostic to the noise distribution. This OT-based framework both helps recovering the causal graph up to Markov equivalence, and offers a
coherent framework for evaluating scores assessing the validity of structural assumptions such as additive noise or post-nonlinear models.

Bringing optimal transport framework to the fore as a valuable toolkit for causal discovery, we believe that this study may pave the way to future works of interest to the causality community. 

\bibliographystyle{plain}
\bibliography{biblio}
\clearpage
\appendix
\begin{center}
    \bfseries\Large Appendix
\end{center}
This appendix is organized as follows.
We present our OT-based version of PC algorithm in Section \ref{apx:pc}.
In Section \ref{apx:brenier}, we briefly review concepts from KR maps and their relation to Brenier maps for the sake of comprehensiveness and their subsequent utilization in our proofs.
We present the proofs of our main results in Section \ref{apx:proofs}.
Section \ref{apx:pnl} includes the derivation of PNLloss based on Lemma \ref{lem:PNLtest} in the text.
Section \ref{apx:numerical_experiments} is devoted to further experimental results, along with details of the experiments included in the main text.
\section{OT-based PC Algorithm}\label{apx:pc}

As discussed in the main text, any constraint-based causal discovery algorithm can be modified to utilize the proposed OT-based conditional independence test.
For illustration purposes, we present an OT-based version of PC algorithm \cite{spirtes2000causation} in this section.
The pseudo-code is provided as Algorithm \ref{alg:pc}.
\paragraph{High-level description.} 
Like classic PC, the algorithm begins with a complete undirected graph.
It keeps track of a counter $\Delta$, increasing it by one at each iteration.
At each iteration, subsets $\mathbf{Z}$ of size $\Delta+2$ are chosen, and SING \cite{morrison2017beyond} is called as a subroutine using only the samples corresponding to variables in $\mathbf{Z}$. 
The output of this subroutine is the matrix $\Omega^\mathbf{Z}$, with entries as defined in Eq.~\eqref{eq:def_omega}.
As long as the entry $\Omega_{k\ell}^\mathbf{Z}$ of the matrix does not exceed the threshold $\tau_{k\ell}^\mathbf
Z$, we conclude that  $X_k$ and $X_\ell$ are conditionally independent with respect to $\pi$, which results in removing the corresponding edge from $\hat{\mathcal{G}}$.
The corresponding separating set $\mathbf{Z}\setminus\{X_k,X_\ell\}$ is stored for orienting the v-structures at the end of the algorithm.
The algorithm iterates as long as the maximum degree of the remaining graph $\hat{\mathcal{G}}$ is at least $\Delta$.
Finally, Meek rules \cite{meek1995causal} are applied to output the essential graph.

\begin{algorithm}
\caption{PC--OT}
\label{alg:pc}
\begin{algorithmic}[1]
    \Statex\textbf{input:} $n$ i.i.d. samples from the observational distribution $\{\mathbf{x}^i\}_{i=1}^n\sim\pi$
    \Statex\textbf{output:} essential graph corresponding to the causal DAG $\mathcal{G}$
    \Function{PCOT}{$\{\mathbf{x}^i\}_{i=1}^n$}
        \State $\hat{\mathcal{G}}\gets$ complete undirected graph on $\mathbf{X}$, $\quad\Delta\gets 0$
        \For{every $k\neq\ell\in\{1,\dots,d\}$} $\:$ SepSet $(X_k,X_\ell)\gets$ \textrm{null}
        \EndFor
        \While{True}
            \For{every subset $\mathbf{Z}\subseteq\mathbf{X}$ of size $\Delta+2$}
            \State $\Omega^\mathbf{Z} \gets$ \Call{SING}{$\{\mathbf{x}_\mathbf{Z}^i\}_{i=1}^n\sim\pi_\mathbf{Z}$} 
                \For{each pair $\{X_k,X_\ell\}\subseteq\mathbf{Z}$}
                    \If{$\Omega_{k\ell}^\mathbf{Z}<\tau_{k\ell}^\mathbf{Z}$}
                        \State delete the edge between $X_k$ and $X_\ell$ in $\hat{\mathcal{G}}$
                        \State SepSet $(X_k,X_\ell)\gets\mathbf{Z}\setminus\{X_k,X_\ell\}$
                    \EndIf
                \EndFor
            \EndFor
            \State $\Delta\gets \Delta +1$
            \If{$\Delta > \mathrm{maxdegree}(\hat{\mathcal{G}})$}                \textbf{break}
            \EndIf
        \EndWhile
        \For{every triplet $k,\ell, m\in\{1,\dots,d\}$}
            \If{$\exists$ an edge between $X_k$ and $X_\ell$, and $X_\ell$ and $X_m$, no edge between $X_k$ and $X_m$}
                \State Orient $X_k\to X_\ell$ and $X_\ell\gets X_m$ in $\hat{\mathcal{G}}$ if and only if $X_\ell\notin$ SepSet $(X_k,X_m)$
            \EndIf
        \EndFor
        \State Apply Meek rules on $\hat{\mathcal{G}}$ \cite{meek1995causal}
        \State\Return $\hat{\mathcal{G}}$
    \EndFunction
\end{algorithmic}
\end{algorithm}

\section{On Knothe-Rosenblatt transport maps}\label{apx:brenier}

In this short part we give details on construction of Knothe-Rosenblatt transport maps between two distributions $\mu$ and $\nu$ in $\dR^d$. For a general recap on transport maps, we refer to Section \ref{sec:transport_maps}. In the following we assume that $\mu$ and $\nu$ have positive densities on $\dR^d$, with respect to Lebesgue measure. These assumptions are made for the sake of simplicity, but can be further relaxed \cite{rieger2012monotonicity,carlier2010from,knothe1957contributions,rosenblatt1952remarks,bonnotte2013}.

A fundamental building block for KR maps is given in this Lemma, which characterizes the one-dimensional monotone transport maps:
\begin{lemma}[See \cite{carlier2010from} and Proposition 2.5 of \cite{rieger2012monotonicity}]
\label{lem:uniqueness_brenier}
    When $\mu$ and $\nu$ are one-dimensional ($d=1$), there exists a unique strictly increasing transport map $T$ from $\mu$ to $\nu$, given by $T := F_{\nu}^{-1}  \circ F_{\mu}$, where $F_{\mu}$ (resp. $F_{\nu}$) is the cumulative distribution function (c.d.f.) of distribution $\mu$ (resp. of $\nu$).
    This map will be referred to as the \emph{(one-dimensional) Brenier map} between distributions $\mu$ and $\nu$.
\end{lemma}

We are now ready to describe the construction of KR maps. Recall that these maps are of the following form
\begin{equation*}
    S(\mathbf{x})=
    \begin{bmatrix*}[l]
        S_1(x_1)\\
        S_2(x_1,x_2)\\
        \vdots\\
        S_d(x_1,\dots,x_d)
    \end{bmatrix*},
\end{equation*} with $S_k$ is strictly increasing in the last variable for all $k$.

Let $(X_1, \ldots,X_d) \sim \mu$ and $(Y_1, \ldots,Y_d) \sim \nu$. 
The KR map $S$ is built recursively. First, let $S_1$ be the (unique) Brenier map from the distribution of $X_1$ to that of $Y_1$. Then, when $S_1, \ldots, S_{k-1}$ are already constructed, we define $S_{k}$ as follows. For every fixed $x_1, \ldots, x_{k-1}$, the map $S_{k}(x_1, \ldots, x_{k-1}, \cdot)$ is defined as the (unique) Brenier map from the distribution of 
$$ (X_k \cond X_{k-1}=x_{k-1}, \ldots, X_1=x_1)$$
to that of
$$(Y_k \cond Y_{k-1} = S_{k-1}(x_1, \ldots, x_{k-1}), \ldots, Y_1 = S_{1}(x_1)) \, .$$
It can be easily checked that this map $S$ previously defined is $(i)$ transporting $\mu$ onto $\nu$ and $(ii)$ satisfying the lower-triangular and monotonicity properties. 

\section{Proofs}\label{apx:proofs}
\thmidwithkr*
\begin{proof}
First note that statement $(ii)$ follows from the identifiablity assumption. 

We prove $(i)$ recursively. Without loss of generality we assume that $\sigma$ is the identity permutation $\id$ and denote $S=S(\sigma)=S(\id)$. For any random variable $Y$, $F_Y$ will denote its cumulative distribution function (c.d.f). 
By definition, transport map $S$ has the following form
\[S(u_1,u_2,\ldots,u_d) = \begin{bmatrix*}[l]
        S_1(u_1)\\
        S_2(u_1,u_2)\\
        \vdots\\
        S_d(u_1,\dots,u_d)
    \end{bmatrix*}
    \, .\]
By definition of a compatible ordering, $X_{1}$ has no parent in $\cG$, hence $X_{1} := f_{1}(U_1)$. The map $S_1$ is by definition (see Appendix \ref{apx:brenier}) the monotone Brenier map between the distribution of $U_1$ and that of $X_1 = f_{1}(U_{1})$. Since by assumption $f_1$ and $F_{U_1}$ are strictly increasing, then $F_{f_{1}(U_{1})}$ is invertible and this Brenier 1D map is given by
\[S_1(u_1) = F_{f_{1}(U_{1})}^{-1} \circ F_{U_1}(u_1) = (F_{U_{1}} \circ f_{1}^{-1})^{-1} \circ F_{U_1}(u_1) =f_1 (u_1) \, .
 \]
Then, assume that \eqref{eq:th:id_with_KR} holds for $1 \leq k < \ell$. Fix $u_1, \ldots, u_{\ell-1} \in \dR$. 
By definition again, $u_\ell \mapsto S_\ell(u_1, \ldots, u_{\ell-1}, u_\ell)$ is the Brenier map between the first marginal distribution of $$(U_\ell \cond U_{\ell-1} = u_{\ell-1}, \ldots, U_1 = u_1) \overset{(d)}{=} U_\ell$$ and that of $$(X_\ell \cond X_{\ell-1} = S_{\ell-1}(u_1, \ldots, u_{\ell-1}), \ldots, X_1 = S_{1}(u_1)) \overset{(d)}{=} f_{\ell}\left( (S_{k}(u_{1}, \ldots, u_{k}))_{k : X_{k} \in \Par(X_{\ell})}, u_{\ell}  \right) ,$$
The second equality in distribution being justified by the fact that $\sigma = \id$ is a compatible ordering. 
Since by assumption $f_\ell$ is strictly increasing in the last variable and $F_{U_\ell}$ is strictly increasing, then $F_{f_{\ell}\left( (S_{k}(u_{1}, \ldots, u_{k}))_{k : X_{k} \in \Par(X_{\ell})}, U_{\ell}  \right) } = F_{ U_{\ell}} \circ f_{\ell}^{-1}\left( (S_{k}(u_{1}, \ldots, u_{k}))_{k : X_{k} \in \Par(X_{\ell})}, \cdot  \right) $ is invertible and this Brenier 1D map is given by
\begin{flalign*}
    S_\ell(u_1, \ldots, u_{\ell-1}, u_\ell) & = F_{f_{\ell}\left( (S_{k}(u_{1}, \ldots, u_{k}))_{k : X_{k} \in \Par(X_{\ell})}, U_{\ell}  \right) } ^{-1} \circ F_{U_\ell}(u_\ell) \\
    & = \left[F_{ U_{\ell}} \circ f_{\ell}^{-1}\left( (S_{k}(u_{1}, \ldots, u_{k}))_{k : X_{k} \in \Par(X_{\ell})}, \cdot  \right) \right] ^{-1} \circ F_{U_\ell}(u_\ell) \\
    & = f_{\ell}^{-1}\left( (S_{k}(u_{1}, \ldots, u_{k}))_{k : X_{k} \in \Par(X_{\ell})}, u_\ell  \right) \, .
\end{flalign*} 
\end{proof}

\lemcitest*
\begin{proof}
    Suppose the independence holds.
    Then the marginal density factorizes as follows.
    \[
        \pi_\mathbf{Z} = \pi_{\mathbf{Z}\setminus\{X_k,X_\ell\}}\cdot\pi_{X_k\vert\mathbf{Z}\setminus\{X_k,X_\ell\}}\cdot\pi_{X_\ell\vert\mathbf{Z}\setminus\{X_k,X_\ell\}},
    \]
    and therefore,
    \begin{equation}\label{eq:factor}
        \log(\pi_\mathbf{Z}) = \log(\pi_{\mathbf{Z}\setminus\{X_k,X_\ell\}})+\log(\pi_{X_k\vert\mathbf{Z}\setminus\{X_k,X_\ell\}})+\log(\pi_{X_\ell\vert\mathbf{Z}\setminus\{X_k,X_\ell\}}).
    \end{equation}
    It is clear from Eq.~\ref{eq:factor} that 
    $\dfrac{\partial^2\log\pi_{\mathbf{Z}}}{\partial x_k\partial x_\ell}=0 \mbox{ on $\dR^{\vert\mathbf{Z}\vert}$}$.

    For the opposite direction, note that the general solution to the PDE $\dfrac{\partial^2\log\pi_{\mathbf{Z}}}{\partial x_k\partial x_\ell}=0 \mbox{ on }\dR^{\vert\mathbf{Z}\vert}$ 
    is given by $\log(\pi_\mathbf{Z})(\mathbf{z})=f(\mathbf{z}\setminus\{x_k\}) + g(\mathbf{z}\setminus\{x_\ell\})$, for some functions $f$ and $g$.
    The marginal density $\pi_\mathbf{Z}$ is then of the form
    \[\pi_\mathbf{Z}(\mathbf{z}) = \exp(f(\mathbf{z}\setminus\{x_k\}))\exp(g(\mathbf{z}\setminus\{x_\ell\})).\]
    Relying on the positivity of the density, we can compute the conditional density as follows: 
    \begin{flalign*}
        \pi_{X_k,X_\ell\vert\mathbf{Z}\setminus\{X_k,X_\ell\}}(\mathbf{z}) &= 
        \frac{\pi_\mathbf{Z}(\mathbf{z})}{\pi_{\mathbf{Z}\setminus\{X_k,X_\ell\}}(\mathbf{z})}
        =\frac{\exp(f(\mathbf{z}\setminus\{x_k\}))\exp(g(\mathbf{z}\setminus\{x_\ell\}))}{\iint e^{f(\mathbf{z}\setminus\{x_k\})} e^{g(\mathbf{z}\setminus\{x_\ell\})} \mathrm{d}x_k \mathrm{d}x_\ell}\\
        &= 
        \frac{\exp(f(\mathbf{z}\setminus\{x_k\}))\exp(g(\mathbf{z}\setminus\{x_\ell\}))}{\int e^{f(\mathbf{z}\setminus\{x_k\})} \mathrm{d}x_k \int e^{g(\mathbf{z}\setminus\{x_\ell\})}\mathrm{d}x_\ell}\\
        &= 
        \frac{ \exp(f(\mathbf{z}\setminus\{x_k\})) \int e^{g(\mathbf{z}\setminus\{x_\ell\})}\mathrm{d}x_\ell}{\int e^{f(\mathbf{z}\setminus\{x_k\})} \mathrm{d}x_k \int e^{g(\mathbf{z}\setminus\{x_\ell\})}\mathrm{d}x_\ell} \cdot  
        \frac{\exp(g(\mathbf{z}\setminus\{x_\ell\})) \int e^{f(\mathbf{z}\setminus\{x_k\})} \mathrm{d}x_k}{\int e^{f(\mathbf{z}\setminus\{x_k\})} \mathrm{d}x_k \int e^{g(\mathbf{z}\setminus\{x_\ell\})}\mathrm{d}x_\ell}\\
        &=\frac{\pi_{\mathbf{Z}\setminus\{X_\ell\}}(\mathbf{z})}{\pi_{\mathbf{Z}\setminus\{X_k,X_\ell\}}(\mathbf{z})}\cdot\frac{\pi_{\mathbf{Z}\setminus\{X_k\}}(\mathbf{z})}{\pi_{\mathbf{Z}\setminus\{X_k,X_\ell\}}(\mathbf{z})}
        = \pi_{X_k\vert\mathbf{Z}\setminus\{X_k,X_\ell\}} (\mathbf{z})\cdot \pi_{X_\ell\vert\mathbf{Z}\setminus\{X_k,X_\ell\}}(\mathbf{z}),
    \end{flalign*} which implies the desired conditional independence relation.
\end{proof}

    

\lemanmtest*
\begin{proof}
The first direction is proved in the main text, applying Theorem \ref{th:id_with_KR} and considering the form of the KR map in equation \eqref{eq:ANM_KR_map}.
For the other direction, without loss of generality we assume that $\sigma$ is the identity permutation $\id$, the proof being identical for any other permutation.
The general solution to the PDE 
\[\frac{\partial}{\partial x_k} B_k \circ S_k (x_1, \ldots, x_k) - 1 = 0\]
is 
\begin{equation}\label{eq:sol_PDE_ANM}
    B_k \circ S_k (x_1, \ldots, x_k) = x_k - h_k(x_1,\ldots,x_{k-1}),
\end{equation}
for some function $h_k:\dR^{k-1}\to \dR$. By definition of transport map $S_k$, $S_k (X_1, \ldots, X_k)$ is a standard Gaussian variable. By \eqref{eq:sol_PDE_ANM}, denoting $U_k := - B_k \circ S_k (X_1, \ldots, X_k)$, we have for all $k$, $X_k = h_k(X_1,\ldots,X_{k-1}) + U_k$. By independence of the Gaussian marginals, the $U$ variables are independent.   
\end{proof}

\lempnltest*
\begin{proof}
Here again, without loss of generality we assume that $\sigma$ is the identity permutation $\id$, the proof being identical for any other permutation.

For the first direction, applying Theorem \ref{th:id_with_KR}, the KR map $S$ from $\pi$ to $\eta$ is of the form 
\vspace{-0.2cm} 
\begin{equation}\label{eq:PNL_KR_map}
S_k(x_1, \ldots, x_1) = M_k \left(h_k^{-1}(x_k) - g_k((x_\ell)_{X_\ell \in \Par(X_{k})})\right), 
\end{equation} where $M_k$ is the strictly increasing transport map from the distribution of $U_{k}$ to a standard Gaussian $\cN(0,1)$. With $B_k := M_k^{-1}$, $S_k$ is thus a solution to PDE \eqref{eq:lem:PNLtest}.

For the other direction, let us assume that for all $1 \leq \ell < k$, 
\begin{equation}\label{eq:PDE_PNL}
    \frac{\partial^2}{\partial x_k \partial x_\ell} B_k \circ S_k = 0\, .
\end{equation}
Applying \eqref{eq:PDE_PNL} for $\ell=1$, this implies that for all $B_k \circ S_k $ is of the form $B_k \circ S_k (x_1, \ldots, x_k) = a_1(x_2, \ldots, x_k) - b_1(x_1,\ldots,x_{k-1})$, where $a_1$ again satisfies \eqref{eq:PDE_PNL} for $\ell=2$, which again implies that $a_1$ is of the form
$a_1(x_2, \ldots, x_k) = a_2(x_3, \ldots, x_k) - b_2(x_2,\ldots,x_{k-1})$. Iterating over $1 \leq \ell \leq k-1$, we obtain that $B_k \circ S_k$ is of the form 
\begin{equation}\label{eq:pnlproof}
B_k \circ S_k (x_1, \ldots, x_k) = g(x_k) - h(x_1, \ldots, x_{k-1}) \, .
\end{equation}

By definition, $S_k$ is strictly increasing in $x_k$, and $B_k$ is a strictly increasing transport map.
Therefore, Eq.~\eqref{eq:pnlproof} implies that $g$ is strictly increasing in $x_k$, and $g^{-1}$ is well-defined.
By definition of transport map $S_k$, $S_k (X_1, \ldots, X_k)$ is a standard Gaussian variable. 
By \eqref{eq:sol_PDE_ANM}, denoting $U_k := - B_k \circ S_k (X_1, \ldots, X_k)$, we have for all $k$, $g(X_k) = h(X_1,\ldots,X_{k-1}) + U_k$. By independence of the Gaussian marginals, the $U$ variables are independent.
Finally, applying the function $g^{-1}$ to both sides, we get
$X_k = g^{-1}(h(X_1,\ldots,X_{k-1})+U_k)$, which is a PNL considering the independence of $U$ variables.

\end{proof}

\section{PNLs}\label{apx:pnl}
Due to space limitations, the derivations of PNLloss was postponed to this appendix.
In view of Lemma \ref{lem:PNLtest},
for a given ordering $\sigma$, we can parameterize each map $B_k(\sigma)$ with vector ${\bm \beta_k}$ as in \eqref{eq:param_beta_k}, and ${\bm \beta^*_k}$ is now estimated by optimizing the loss given by Lemma \ref{lem:PNLtest}:
\begin{flalign}\label{eq:PNL_beta_star_k}
    {\bm \beta^*_k} & := \argmin_{\bm \beta_k}  \sum_{\substack{1 \leq \ell \leq k \\ \ell \neq k}}\dE_{\pi}\left[ \left| \frac{\partial^2}{\partial x_{\sigma(l)}\partial x_{\sigma(k)}} [B_k(\sigma)]_{{\bm \beta_k}} \circ (S(\sigma)_{\bm \alpha^*})_k (X_{\sigma(1)},\ldots,X_{\sigma(k)}) \right| \right] \nonumber \\
    & \approx \argmin_{\bm \beta_k} \,  [\mathrm{PNLloss}_{k}(\sigma,\mathbf{x})]({\bm \beta_k}),
\end{flalign} where 
\begin{equation}\label{eq:PNL_loss_k_beta}
    [\mathrm{PNLloss}_{k}(\sigma,\mathbf{x})]({\bm \beta_k}) := \sum_{\substack{1 \leq \ell \leq k \\ \ell \neq k}} \sum_{i=1}^{n}  \left| \frac{\partial^2}{\partial x_{\sigma(l)}\partial x_{\sigma(k)}} [B_k(\sigma)]_{{\bm \beta_k}} \circ (S(\sigma)_{\bm \alpha^*})_k (x^i_{\sigma(1)},\ldots,x^i_{\sigma(k)}) \right| .
\end{equation}

\paragraph{PNL loss.} The \emph{PNL loss} of an ordering $\sigma$, parameterized by ${\bm \gamma} \in (\dR_{>0})^d$ is defined as
\begin{equation}\label{eq:PNL_loss}
    \mathrm{PNLloss}_{\bm \gamma}(\sigma, \mathbf{x}) := \sum_{k=1}^{d} \gamma_k  [\mathrm{PNLloss}_{k}(\sigma,\mathbf{x})]({\bm \beta_k}) ,
    \end{equation} where ${\bm \alpha^*}$, ${\bm \beta^*_k}$ and $\mathrm{ANMloss}_{k}(\sigma,\mathbf{x})$ are defined in \eqref{eq:alpha_star}, \eqref{eq:PNL_beta_star_k} and \eqref{eq:PNL_loss_k_beta}. 
    
    We note in particular that $$\mathrm{ANMloss}_{\bm \gamma}(\sigma) = 0 \implies \mathrm{PNLloss}_{\bm \gamma}(\sigma) = 0,$$
which agrees with the fact that $\{ \mbox{ANMs} \} \subsetneq \{ \mbox{PNLs} \}$.

\section{Further on numerical experiments}\label{apx:numerical_experiments}
In this section, we first provide comprehensive details of the numerical experiments included in the main text.
Subsequently, we unveil novel numerical experiments, including a numerical experiment with the real-world dataset 'Sachs' \cite{sachs2005causal}.

\subsection{Details of the experiments in the text}
\paragraph{PC-OT experiments.}
These experiments were conducted based on the following SEM:
\begin{equation}\label{eq:sem_num_exp}
  \begin{split}
X_1 & := U_1\\
X_2 & := U_2\\
X_3 & := X_1^2 + X_2 + U_3\\
X_4 &  := U_4 \\
X_5 &  := 0.5 X_1^2 - 0.5 X_4^2 + X_1X_4 + U_5\\
X_6 &  := X_4^3 - X_5 + U_6
  \end{split}
\quad \mbox{with} \quad  \begin{split}
U_1 & \sim 0.2 \mathcal{N}(0,1) \times \mathcal{N}(0,1)\\
U_2 & \sim  (\mathrm{Gumbel}(0,0.7)-2.5)/2.5 \\
U_3 & \sim  \mathcal{N}(0,1) \times \mathrm{Exp}(1)/8\\
U_4 &  \sim (\mathrm{Ber}(1/2) \times \mathrm{Exp}(1)-3)/2  \\
U_5 &  \sim (\mathrm{Ber}(1/2) \times \Gamma(2,3)-24)/12\\
U_6 & \sim  \mathrm{Gumbel}(0,0.5)-1.5 \\
  \end{split}
\end{equation} The underlying causal graph $\cG$ is given by Figure \ref{fig:exp-pcot}.
The overall loss of Figure \ref{fig:overall} was defined as the sum of number of missing, extra, and misoriented edges.
Figure \ref{fig:PCOT_vs_PC_GS2} illustrates the decomposition of these error terms.
Note that the comparison between the number of misoriented edges was included in Figure \ref{fig:misorientedmain}.

\begin{figure}[ht]
	\centering
    \hspace{-1cm}
\begin{subfigure}[c]{0.45\textwidth}
        \includegraphics[scale=0.44]{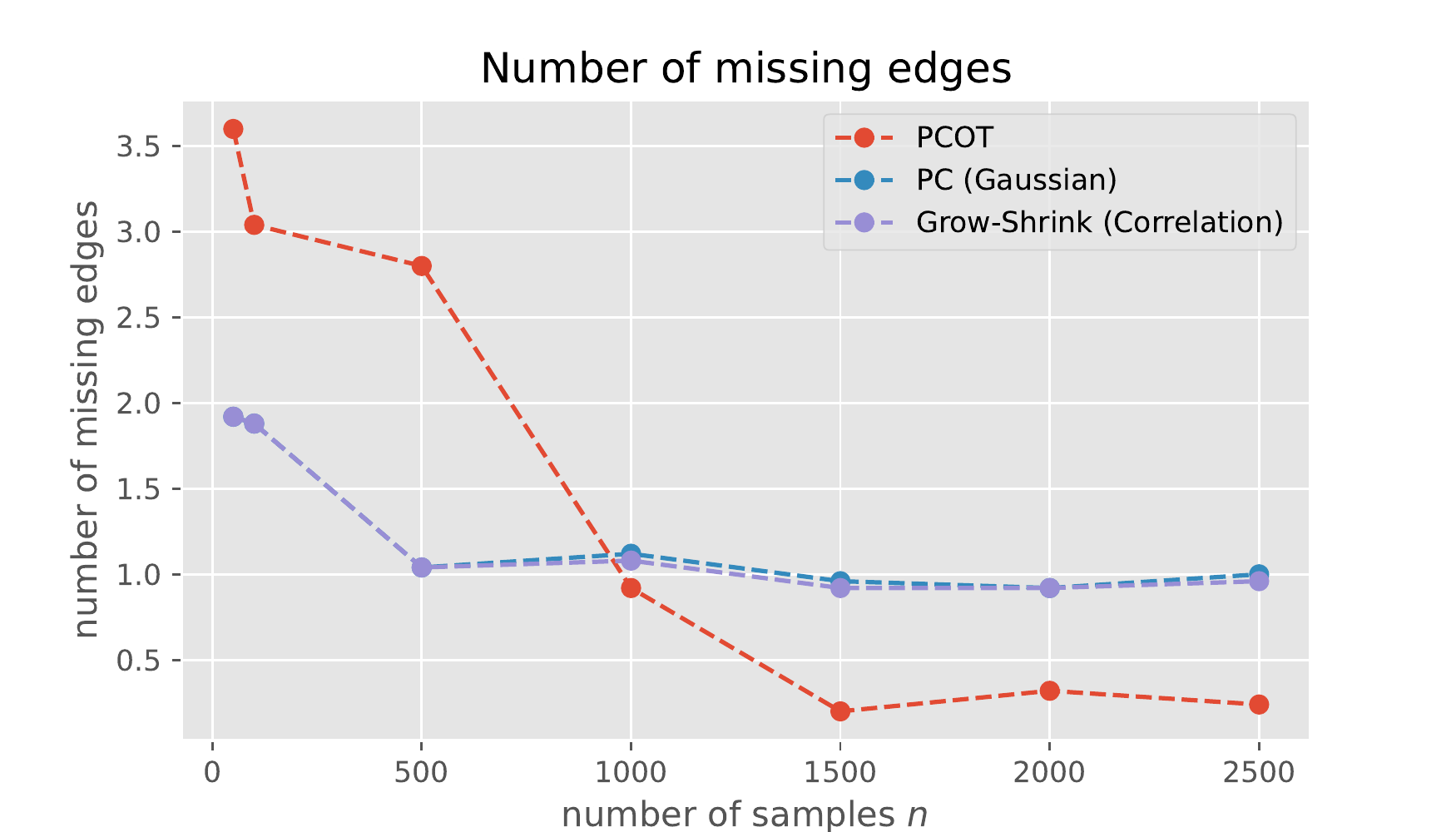}
  \subcaption{average number of missing edges}
	\end{subfigure}
 \hspace{1.5cm}
 \begin{subfigure}[c]{0.45\textwidth}
		\includegraphics[scale=0.44]{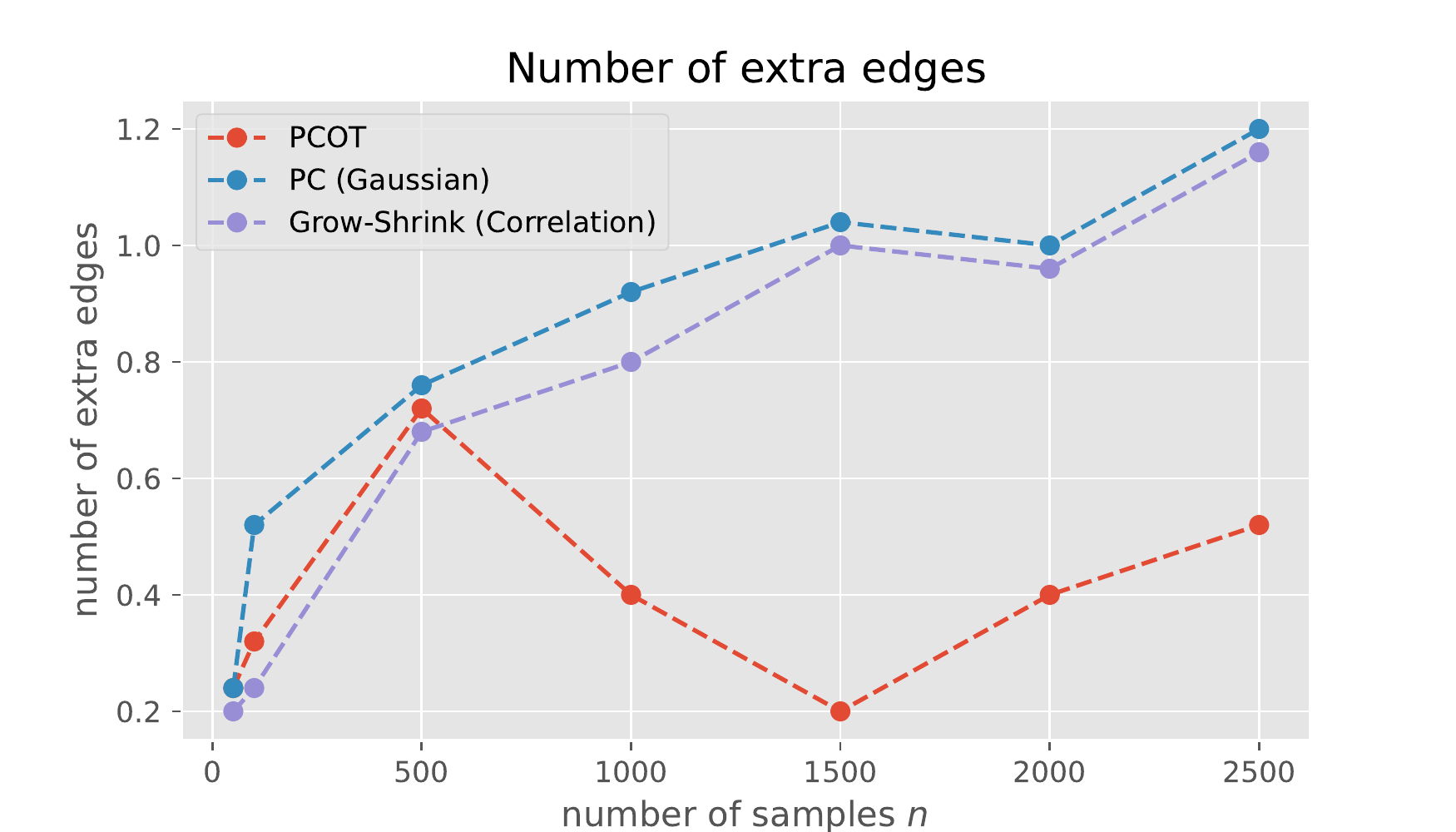}
  \subcaption{average number of extra edges}
	\end{subfigure}
\caption{Decomposition of the errors made by PC-OT, PC (Gaussian) and Grow-Shrink.}
\label{fig:PCOT_vs_PC_GS2}
\end{figure}

\paragraph{ANM experiments.}
Within these experiments, we worked with the following SEM:
\begin{equation}\label{eq:sem_num_exp2}
  \begin{split}
X_1 & := U_1\\
X_2 & := 0.5X_1^2+U_2\\
X_3 & := \log(X_1^2) + U_3\\
X_4 &  := 2X_2(X_2+1)+ U_4 \\
X_5 &  := 0.5 X_1^2 - 0.5 X_4^2 + X_1X_4 + U_5\\
X_6 &  := 0.25X_3^2 - X_5 + U_6
  \end{split}
\quad \mbox{with} \quad  \begin{split}
U_1 & \sim 0.2 \mathcal{N}(0,1)^2\\
U_2 & \sim  0.5 \mathcal{N}(-2.5,1) \\
U_3 & \sim  \log(\mathcal{N}(0,1)^2+1)\\
U_4 &  \sim  0.3\mathcal{N}(0,1)^2 \\
U_5 &  \sim \log(\mathcal{N}(0,1)^2+1)\\
U_6 & \sim   \mathcal{N}(0,1)\\
  \end{split}
\end{equation} The underlying causal graph $\cG$ is given by Figure \ref{fig:exp-anm}.
For each sample size, we repeated the experiment $20$ times, and the box plots of the ANMlosses corresponding to each permutation was depicted in Figure \ref{fig:ANM-OT_experiments}.
The permutation corresponding to the true causal order was $1\to2\to 4\to5\to3\to6$, which had the lowest ANMloss among all compatible permutations.
Further, the gap between the ANMlosses increased as the number of samples grew larger.

\subsection{Further experiments}
\paragraph{Real-world data.}
In this section, we consider a dataset corresponding to the causal relations among components of a cellular signaling network based on single-cell data, namely 'Sachs' dataset \cite{sachs2005causal}.
This dataset comprises samples of $7446$ primary human immune system cells.
We consider a subnetwork of this dataset corresponding to the proteins $Pcl_\gamma$, $PIP3$, $PIP2$, $PKC$ and $Akt$.
The causal mechanisms between these proteins are depicted in Figure \ref{fig:proteins}.

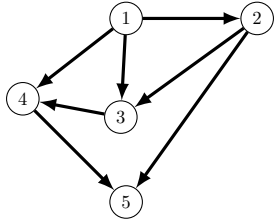
\begin{figure}[h]
    \centering
        \begin{tikzpicture}[scale=0.6, every node/.style={scale=0.7}]
			\tikzset{vertex/.style = {shape=circle,draw,minimum size=1em}}
			\tikzset{edge/.style = {->,very thick,> = latex,sibling distance=20mm}}
			\node[vertex] (a) at (0,0) {$1$};
			\node[vertex] (b) [right  = 1.3cm of a] {$2$};
			\node[vertex] (c) [below left = 1cm and 1.5cm of b] {$3$};
			\node[vertex] (e) [below  = 2cm of a] {$5$};
			\node[vertex] (d) [above left = 1.5cm of e] {$4$};
			\draw[edge] (a) to (b);
			\draw[edge] (a) to (c);
            \draw[edge] (a) to (d);
            \draw[edge] (b) to (c);
            \draw[edge] (b) to (e);
            \draw[edge] (c) to (d);
            \draw[edge] (d) to (e);
		\end{tikzpicture}
    \caption{Causal mechanisms pertaining to the proteins $1\coloneqq Pcl_\gamma$, $2\coloneqq PIP3$, $3\coloneqq PIP2$, $4\coloneqq PKC$ and $5\coloneqq Akt$.}
    \label{fig:proteins}
\end{figure}

Since the dataset comprises values between $1.0$ and $9058$, we applied a logarithm function so that the support spans the real numbers.
We then provided the Markov equivalence class of Figure \ref{fig:proteins} (which consists of $10$ different DAGs) to our ANM-OT algorithm.
Table \ref{tab:anm} below demonstrates the ANMlosses corresponding to each permutation.
The ground truth permutation is the last entry of the table, $1\to 2\to 3\to 4\to 5$.
As can be seen in Table \ref{tab:anm}, the ground truth permutation has the second lowest ANMloss, following the permutation $2\to 1\to 3\to 4\to 5$, which is a transposition of the true permutation.

\renewcommand{\arraystretch}{1.2}
\begin{table*}[ht]
\caption{ANMlosses pertaining to the permutations compatible with the Markov equivalence class of the DAG in Figure \ref{fig:proteins}.}
    \centering
    \begin{tabular}{ c| c}
        \toprule
        Permutation & ANMloss (ANM-OT)\\
        \hline
        $3\to 2\to 1\to 4\to 5$ &46,310.82\\ \hline
        $2\to 3\to 1\to 4\to 5$ &36,563.22\\ \hline
        $2\to 1\to 3\to 4\to 5$ &28,517.74\\ \hline
        $4\to 3\to 1\to 2\to 5$&48,449.60\\\hline
        $3\to 4\to 1\to 2\to 5$&49,686.34\\\hline
        $3\to 1\to 4\to 2\to 5$&49,448.33\\\hline
        $4\to 1\to 3\to 2\to 5$&38,816.92\\\hline
        $1\to 4\to 3\to 2\to 5$&33,776.88\\\hline
        $1\to 3\to4\to 2\to 5$ &35,402.31\\\hline
        $1\to 2\to 3\to 4\to 5$&31,732.97\\
        \bottomrule
    \end{tabular}
    \label{tab:anm}
\end{table*}

\paragraph{Another illustrative example for PC-OT.}
To illustrate the effectiveness of PC-OT on data with non-Gaussian noise, we provide a numerical experiment on a small model.
The SEM we consider is as follows.
\begin{equation}\label{eq:sembis}
    \begin{split}
X_1 & := U_1/450\\
X_2 & := U_2\\
X_3 & := (X_2^3 + \log(|X_2|X_1^2))/15\\
  \end{split}
\quad \mbox{with} \quad  \begin{split}
U_1 & \overset{(d)}{=} \sqrt{4/3}(\mathrm{Pow}(4)-3/2) \mbox{ conditioned to be } \leq 1000\\
U_2 & \overset{(d)}{=} (\mathrm{Gumbel}(0,0.7)-2.5)/2.5 \\
U_3 & \overset{(d)}{=} \mathrm{Ber}(1/2) \times \mathrm{Exp}(1/2)\\
  \end{split}
\end{equation}

Note that the DAG corresponding to the SEM of Eq.~\eqref{eq:sembis} is a v-structure, namely $X_1\to X_3\gets X_2$.
We repeated the experiments of Section \ref{sec:exp_pc_ot} using the SEM of Eq.~\eqref{eq:sembis}.
For comprehensiveness, we also included a version of PC algorithm provided with a kernel-based CI test, namely 
HSIC-Gamma \cite{gretton2007kernel}
provided in the CDT package \cite{kalainathan2018structural}.
The results are depicted in Figure \ref{fig:resultsbis}.
As witnessed in Figure \ref{fig:resultsbis}, PC-OT performs significantly better than PC with Gaussian CI tests.
Although the performance of PC-OT is comparable to PC with the kernel-based CI tests, the computing time of the kernel-based algorithm appears to be drastically growing with sample size.
In contrast, PC-OT does not suffer from a growing runtime.
It is noteworthy that with $2000$ samples, the kernel-based method necessitates a runtime that is 14 times greater compared to that of PC-OT.
\begin{figure}[ht]
	\centering
    \hspace{-.5cm}
    \begin{subfigure}[c]{0.45\textwidth}
        \includegraphics[scale=0.44]{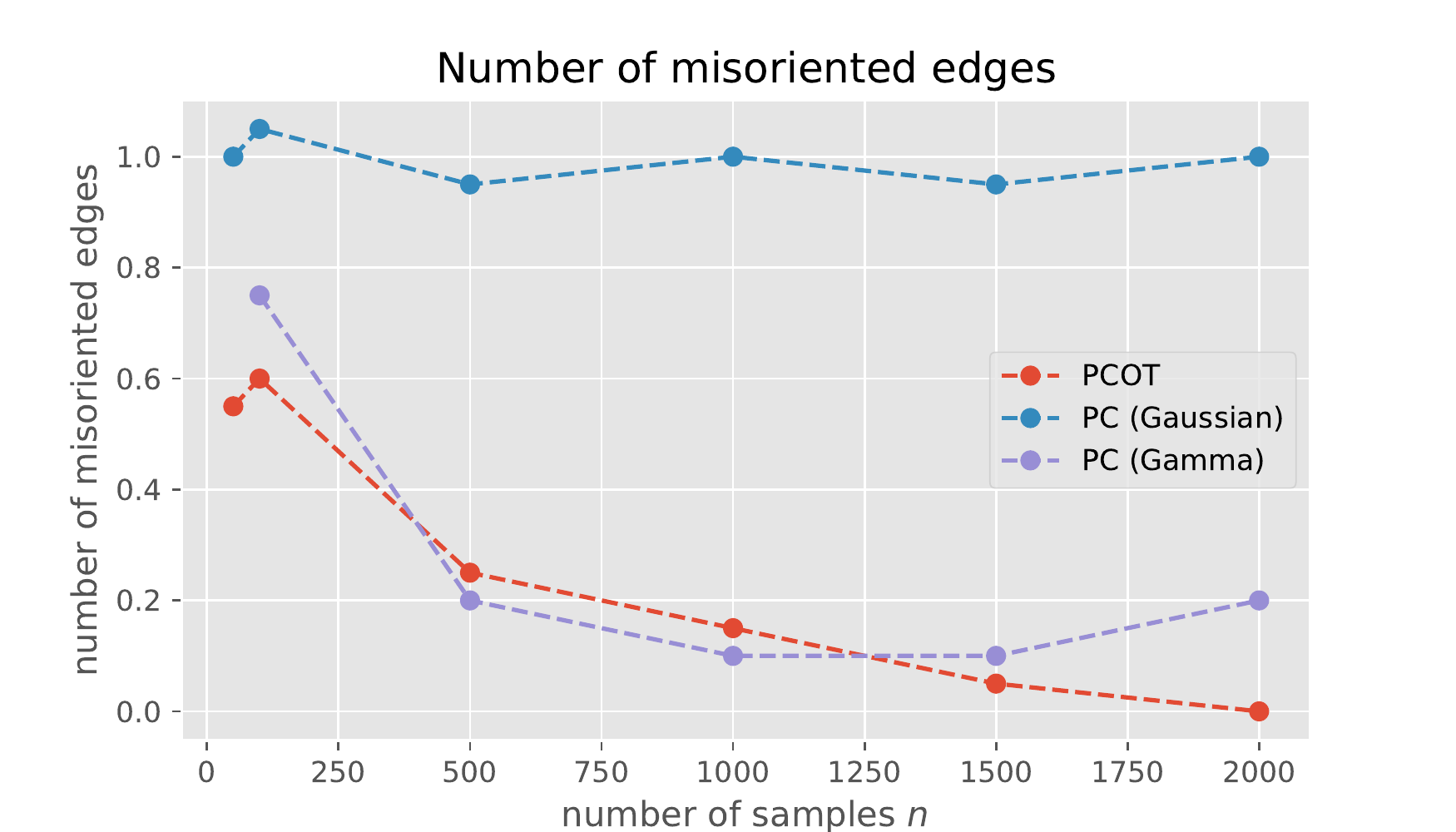}
  \subcaption{average number of misoriented edges}
	\end{subfigure}
 \hspace{1cm}
\begin{subfigure}[c]{0.45\textwidth}
        \includegraphics[scale=0.44]{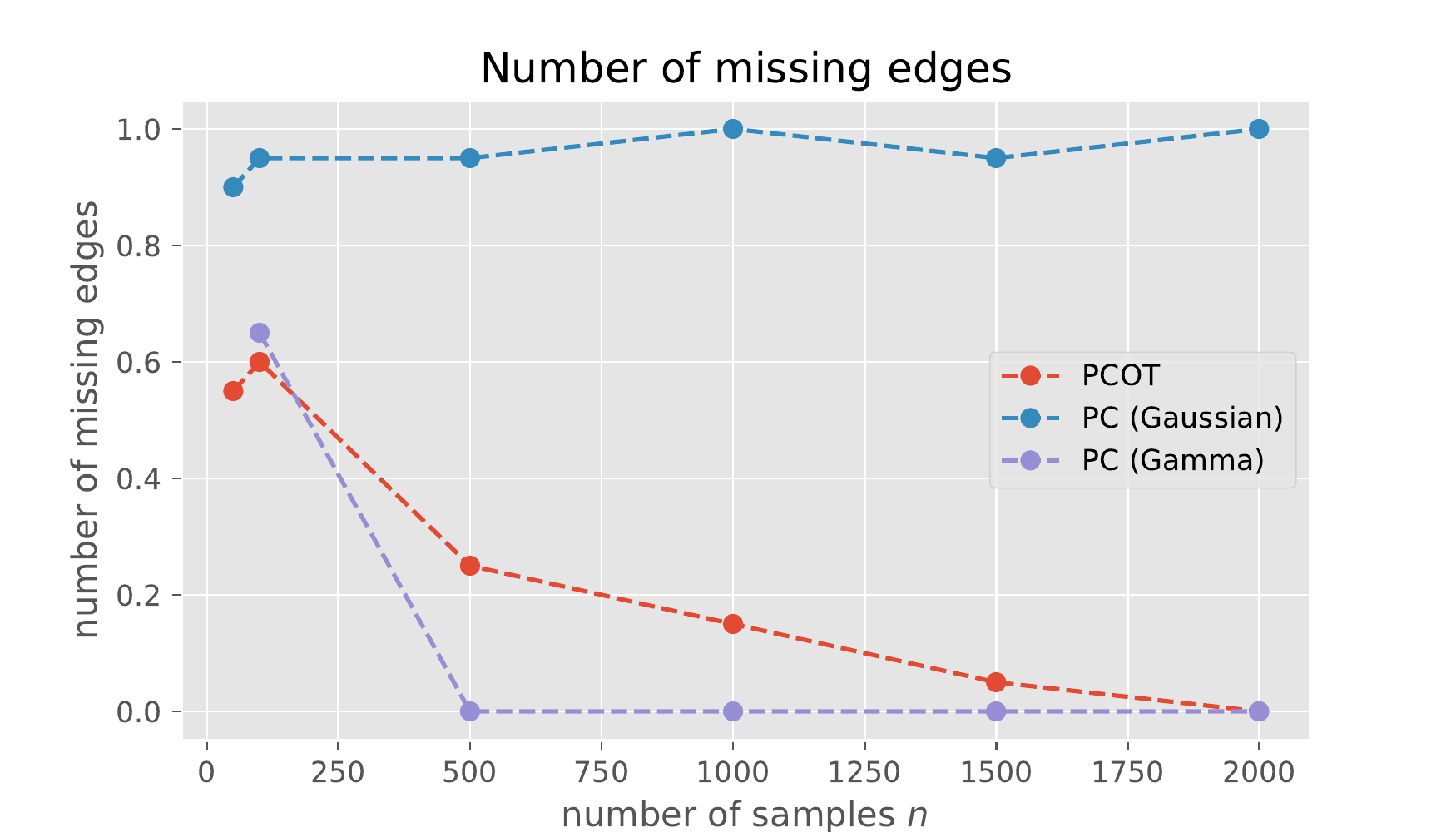}
  \subcaption{average number of missing edges}
	\end{subfigure}
 \hspace{-.5cm}
 \begin{subfigure}[c]{0.45\textwidth}
		\includegraphics[scale=0.44]{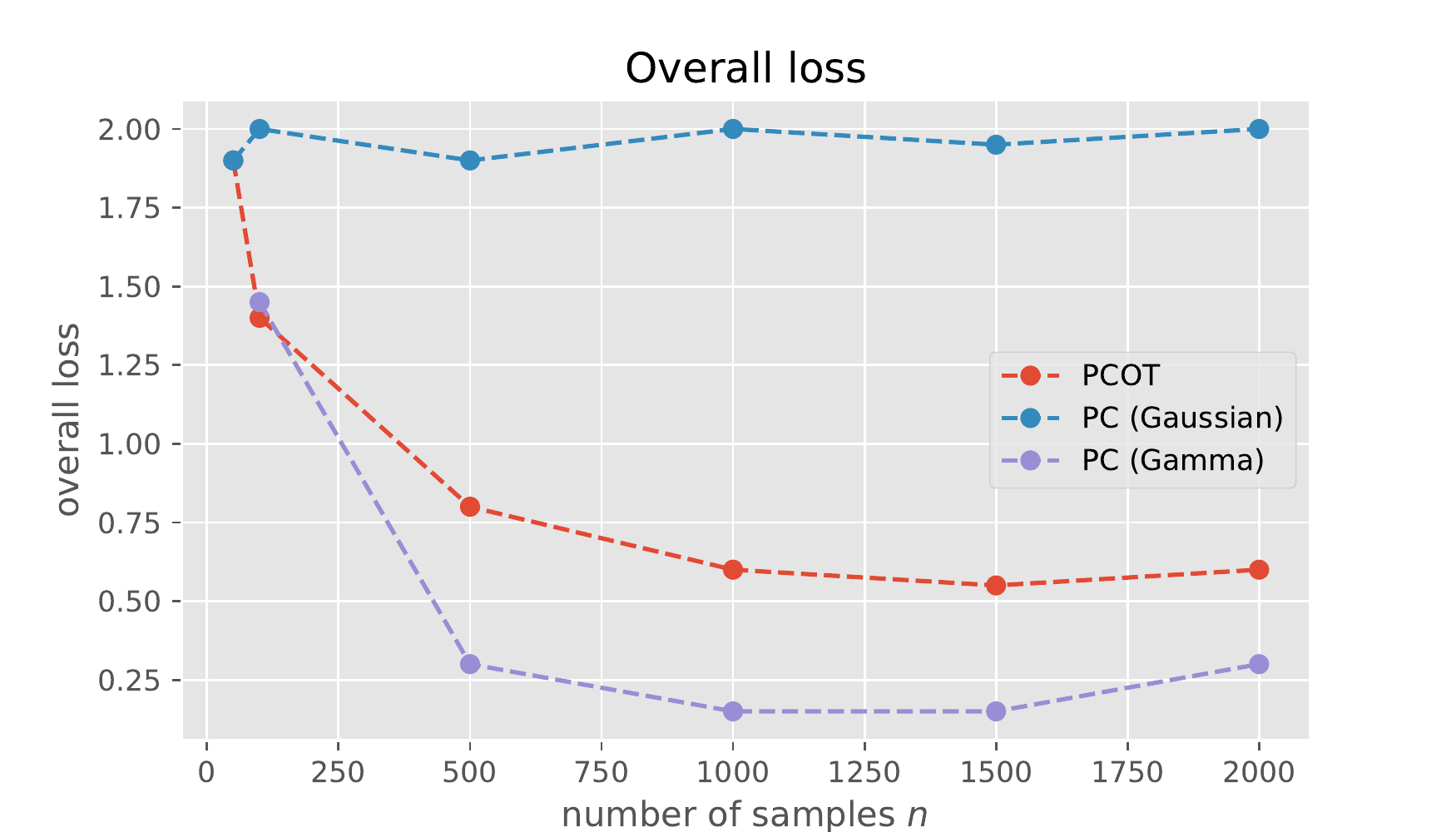}
  \subcaption{overall average loss}
	\end{subfigure}
 \hspace{1cm}
 \begin{subfigure}[c]{0.45\textwidth}
		\includegraphics[scale=0.44]{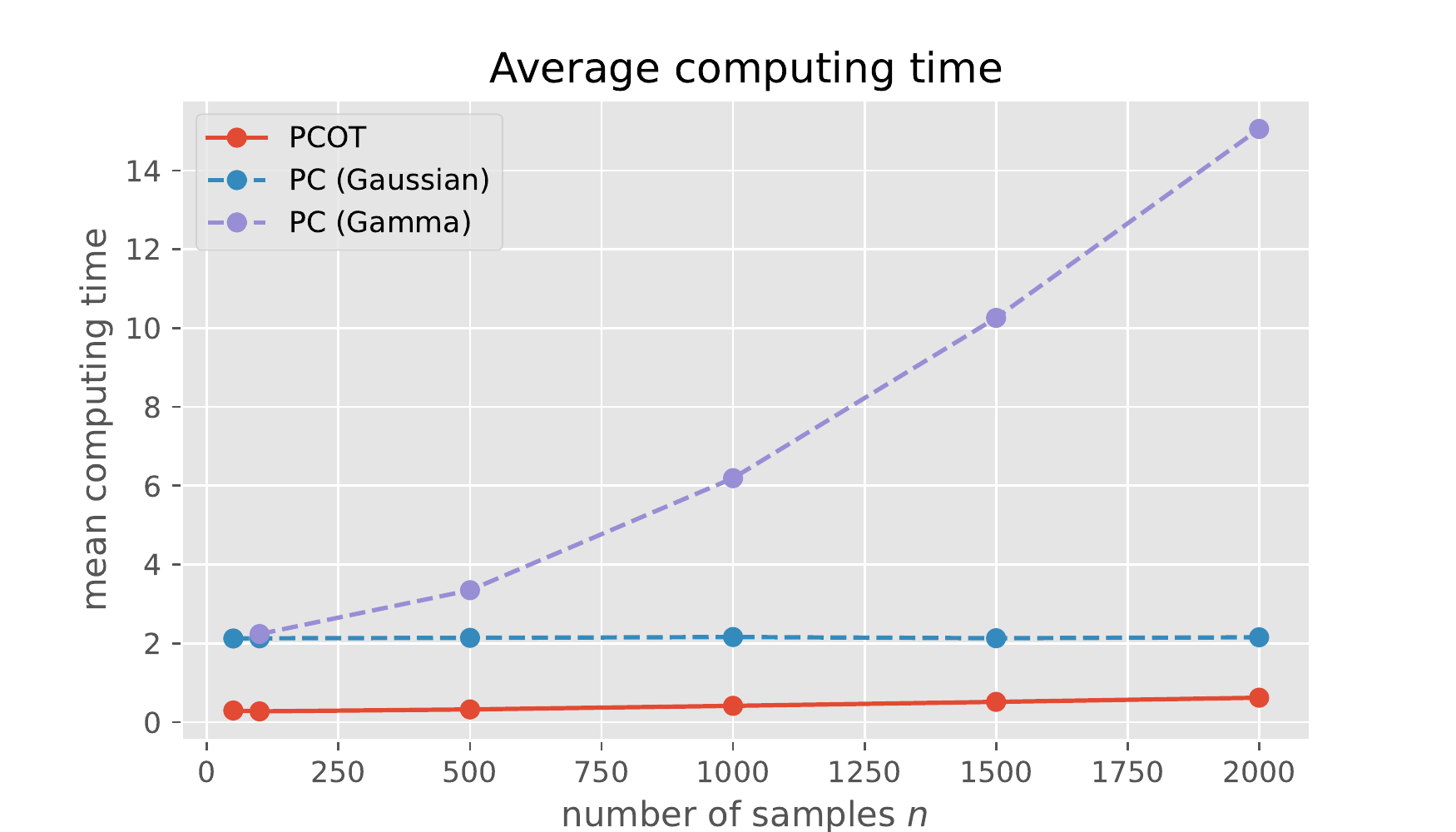}
  \subcaption{average computing time of the algorithms}
	\end{subfigure}
\caption{Performance of PC-OT, PC (Gaussian CI test) and PC (HSIC-Gamma CI test) on the illustrative example with SEM of Eq.~\eqref{eq:sembis}.}
\label{fig:resultsbis}
\end{figure}
\end{document}